\documentclass[onecolumn,notitlepage,superscriptaddress,11pt,tightenlines,nofootinbib]{revtex4-1}    

\usepackage[ colorlinks = true,
             linkcolor = blue,
             urlcolor  = blue,
             citecolor = red,
             anchorcolor = green,
]{hyperref}

\usepackage{amssymb,amsmath,amsthm}
\usepackage[T1]{fontenc}
\usepackage[mathscr]{euscript}
\usepackage{microtype}
\usepackage{subfigure}
\usepackage{enumitem}
\usepackage{graphicx}
\usepackage{overpic}

\theoremstyle{plain}
\newtheorem{lemma}{Lemma}
\newtheorem{proposition}[lemma]{Proposition}
\newtheorem{theorem}[lemma]{Theorem}
\newtheorem{corollary}[lemma]{Corollary}
\newtheorem*{directpart}{Direct Part of Theorem~\ref{th:main}}

\theoremstyle{remark}
\newtheorem{remark}{Remark}
\newtheorem{definition}{Definition}
\newtheorem{example}{Example}

\newcommand{\cX}{\mathscr{X}}
\newcommand{\cY}{\mathscr{Y}}
\newcommand{\cS}{\mathscr{S}}
\newcommand{\cD}{\mathcal{D}}
\newcommand{\cK}{\mathcal{K}}
\newcommand{\cW}{\mathcal{W}}

\newcommand{\cM}{\mathscr{M}}
\newcommand{\cC}{\mathcal{C}}
\newcommand{\cP}{\mathcal{P}}

\newcommand{\cE}{\mathcal{E}}

\newcommand{\bbP}{\mathbb{P}}
\newcommand{\cH}{\mathscr{H}}
\newcommand{\cG}{\mathscr{G}}

\newcommand{\perr}{p_{\textrm{err}}}
\newcommand{\abs}[1]{|#1|}
\newcommand{\eps}{\varepsilon}
\newcommand{\id}{\mathrm{id}}

\newcommand{\proj}[1]{|#1\rangle\!\langle #1|}

\newcommand{\diff}{\textnormal{d}}

\DeclareMathOperator*{\Exp}{E}
\DeclareMathOperator*{\Var}{Var}

\DeclareMathOperator{\tr}{tr}

\renewcommand{\vec}[1]{\textnormal{\textit{\textbf{#1}}}}
\DeclareMathOperator*{\argmax}{\arg\max}
\DeclareMathOperator*{\argmin}{\arg\min}
\DeclareMathOperator{\conv}{conv}

\DeclareMathOperator{\supp}{supp}
\DeclareMathOperator{\im}{im}

\newcommand{\cSo}{\cS_{\circ}}
\newcommand{\ccSo}{\overline{\cS}_{\circ}}
\newcommand{\cXo}{\cX_{\circ}}

\begin{document}

\title{\Large Second-Order Asymptotics for the Classical\\ Capacity of Image-Additive Quantum Channels}

 \author{Marco Tomamichel}                   
 \email{marco.tomamichel@sydney.edu.au}
 \affiliation{School of Physics, The University of Sydney, Australia}
 \affiliation{Centre for Quantum Technologies, National University of Singapore, Singapore}
 
 \author{Vincent Y.~F.~Tan}
  \email{vtan@nus.edu.sg}
  \affiliation{Department of Electrical and Computer Engineering, and \\ Department of Mathematics, National University of Singapore, Singapore}

\begin{abstract}
We study non-asymptotic fundamental limits for transmitting classical information over memoryless quantum channels, i.e.\ we investigate the amount of classical information that can be transmitted when a quantum channel is used a finite number of times and a fixed, non-vanishing average error is permissible. 
In this work we consider the classical capacity of quantum channels that are image-additive, including all classical to quantum channels,
as well as the product state capacity of arbitrary quantum channels.
In both cases we show that the non-asymptotic fundamental limit admits a second-order approximation that illustrates the speed at which the rate of optimal codes converges to the Holevo capacity as the blocklength tends to infinity. The behavior is governed by a new channel parameter, called channel dispersion, for which we provide a geometrical interpretation.
\end{abstract}

\maketitle

\section{Introduction}

One of the landmark achievements in quantum information theory is the establishing of the coding theorem for sending classical information across a noisy quantum channel by Holevo~\cite{holevo98}, and independently by Schumacher-Westmoreland~\cite{schumacher97}\,---\,the so-called HSW theorem. 
The HSW theorem can be formally stated as follows: Let $\cW^n$ denote the $n$-fold memoryless composition of the channel $\cW$ and let
$M^*(\cW^n,\eps)$ denote the maximum size of a length-$n$ block code for the channel $\cW$ with average error probability $\eps \in (0,1)$. Then, the HSW theorem, together with the weak converse  established by Holevo~\cite{holevo73b} in the 1970s (the Holevo bound), asserts that 
\begin{align}
\label{eq:holevo0}
C(\cW) := \lim_{\eps\to 0} \liminf_{n\to\infty} \frac{1}{n}\log M^*(\cW^n,\eps) = \lim_{n \to \infty} \frac{1}{n} \chi(\cW^{n}) ,
\end{align}
where $\chi(\cW)$ is the \emph{Holevo capacity} of the channel. (We define all quantities precisely in the following.) Let us emphasize that the Holevo capacity is generally not additive~\cite{hastings09}, and we can thus not simplify the limit on the right hand side of~\eqref{eq:holevo0} without further assumptions.
 
However, for discrete classical-quantum (c-q) channels, the converse part of  HSW theorem was strengthened significantly by Ogawa-Nagaoka~\cite{ogawa99} and Winter~\cite{winter99,winterthesis} who proved the {\em strong converse} for discrete memoryless c-q channels, namely
\begin{align}
  \label{eq:sc-winter}
 \lim_{n\to\infty} \frac{1}{n}\log M^*(\cW^n,\eps) = \chi(\cW) = C(\cW),\quad\mbox{for all } \eps\in (0,1).
\end{align}
In the work by Ogawa-Nagaoka~\cite{ogawa99}, the strong converse was proved using ideas from Arimoto's strong converse proof~\cite{arimoto73} for classical channels, which itself was based on techniques to prove Gallager's random coding error exponent~\cite{gallager65}. Hence, Ogawa and Nagaoka's proof~\cite{ogawa99} also applies to c-q channels whose inputs are not necessarily discrete. Winter's strong converse proof~\cite{winter99}, on the other hand, is based on the method of types~\cite{csiszar98} which is a powerful tool developed in classical information theory for discrete memoryless systems. Winter then combines this method with a suitable discretization of the output space to show the strong converse for non-stationary channels~\cite{winterthesis}.
We also mention the work by Hayashi-Nagaoka~\cite{hayashi03} in which a necessary and sufficient condition was provided for the strong converse property to hold for general (not only memoryless) c-q channels.
More recently, Wilde-Winter-Yang~\cite{wilde13} established that the strong converse, Eq.~\eqref{eq:sc-winter}, also holds if $\cW$ is an \emph{entanglement-breaking} channel or a Hadamard channel. In particular, this shows that the Holevo capacity is additive for these channels.

In this work we focus our attention on channels $\cW$ that are (tensor product) \emph{image-additive}~\cite{wolf14}, namely quantum channels $\cW$ that satisfy 
\begin{align}
 \im(\cW^n) = \conv\big( \im(\cW)^{\otimes n} \big) , \label{eq:im-add}
\end{align}
where $\im(\cW)$ denotes the image of the channel (i.e.\ the set of all quantum states that can be output by $\cW$ if the input is a quantum state) and $\conv$ denotes the convex hull. This class of channels is a proper subset of the entanglement-breaking channels but strictly larger than c-q channels~\cite{wolf14}. 
Finally, if we restrict the input to an arbitrary quantum channel to product states (or, more generally, separable states), then the respective channel images automatically satisfy~\eqref{eq:im-add}.

We are interested in characterizing $M^*(\cW^n, \eps)$ for these channels beyond the strong converse statement in~\eqref{eq:sc-winter}. This quantity represents the fundamental limit for the size of a codebook that allows transmission of classical information over $n$ uses of the quantum channel $\cW$ up to an error $\eps$. Notably such communication schemes generally require a joint measurement of $n$ quantum systems at the receiver's terminal, which is technologically challenging even for moderate values of $n$. Thus, an asymptotic characterization for $n \to \infty$ as in~\eqref{eq:sc-winter} seems insufficient.
%
To this end, our goal here is to approximate $M^*(\cW^n, \eps)$  in terms of efficiently computable quantities for large but finite $n$. 

For image-additive channels, the results of Wilde-Winter-Yang in fact imply that~\cite{wilde13}
\begin{align}
\log M^*(\cW^n,\eps) = n C(\cW)+ O(\sqrt{n}),\quad\mbox{for all } \eps\in (0,1).
\end{align}
Our present work refines the $O(\sqrt{n})$ term by identifying the implied constant in this remainder term as a function of $\eps$ and a new channel parameter called the \emph{dispersion} of the quantum channel. The resulting second-order approximation generalizes results for classical channels that go back to Strassen's work in the 1962~\cite{strassen62}. In this seminal work, he showed for most well-behaved discrete classical channels $W:\cX\to\cY$ that 
\begin{align}
 \log M^*(W^n,\eps) = nC(W) + \sqrt{nV_\eps (W)}\,\Phi^{-1}(\eps) + O(\log n),  \label{eqn:strass}
\end{align}
where $C(W)$ is the Shannon capacity, $\Phi$ is the cumulative distribution function of a standard normal random variable, and $V_\eps(W)$ is another fundamental property of the channel known as the $\eps$-channel dispersion, a term coined by Polyanskiy {\em et al.}~\cite{polyanskiy10}. Refinements to and extensions of the expansion of $\log M^*(W^n,\eps)$ were pursued by Hayashi~\cite{hayashi09}, Polyanskiy {\em et al.}~\cite{polyanskiy10} and the present authors~\cite{tomamicheltan12}.\footnote{The latter two works establish that the remainder term satisfies $O(\log n) = \frac12 \log n + O(1)$ for most channels, and as such the third-order contribution is independent of the detailed channel description.}

\subsection{Main Contributions}

In Section~\ref{sec:pre} we introduce the necessary concepts and definition required to formally state our main results, which we detail in Section~\ref{sec:results}.
There are three main contributions in this paper: 

\begin{enumerate}[itemsep=0mm]

\item It is a well-known fact that the capacity of a classical or c-q channel can be represented geometrically as the divergence radius of the channel image. In this paper, in the course of proving our  main result, and especially the converse part, we leverage this fact heavily and refine the geometric interpretation of the Holevo capacity in Section~\ref{sec:res-radius}.

\item We develop a one-shot converse bound on $M^*(\cW, \eps)$ in terms of the geometry of  the image of the channel by employing a non-asymptotic quantity known as the {\em $\eps$-hypothesis testing divergence radius}. This is a one-shot analogue of the divergence radius that is commonly used to characterize the channel capacity. We find that such an approach allows to shift our attention  from the input to the output space already in the non-asymptotic (one-shot) regime. Indeed, all the necessary calculations to yield the second-order approximation are done in the output space, thus allowing the input space to be arbitrary. 

This approach of working solely on the output space by employing a one-shot divergence radius to find the converse of the second-order approximation is new 
and does not have a classical analogue.

\item We then use this technique to refine the asymptotic expansion of $\log M^*(\cW^n,\eps)$ for c-q channels whose input alphabet is neither discrete nor otherwise structured. In fact our only requirement is that the image of the channel is comprised of quantum states on a finite-dimensional Hilbert space. We prove a quantum analogue of Strassen's~\cite{strassen62} refinement to the Shannon capacity in \eqref{eqn:strass}. This result is presented as Theorem~\ref{th:main} and discussed in Section~\ref{sec:res-main}. 

Finally, we show how our result for c-q channels with unstructured inputs can be adapted to yield an asymptotic expansion for all image-additive channels as well as the product state capacity of arbitrary quantum channels in Section~\ref{sec:image-add}


\end{enumerate}

Because of the generality that is being afforded in our setup, several auxiliary technical results have to be developed  either by modifying arguments from the literature or proving them from scratch. These results may be of independent interest in other contexts. First, we develop several alternative representations of the  divergence radius that turn out to be amenable for computations involved in both the direct part and converse parts of the proof of our main theorem. Second, in the course of proving the direct part, we also show,  by appealing to Caratheodory's theorem, that it suffices to choose a finite input ensemble in order to achieve the second-order approximation.  Third, for the converse part, to deal with ensembles of ``bad'' states that are not close to Holevo capacity-achieving, we construct an appropriate $\gamma$-net whose size can be controlled appropriately and whose elements serve to approximate those ensembles of ``bad'' states. (Notably, Winter~\cite[Thm.~II.7]{winterthesis} also employed a related idea to get beyond the assumption of discrete input alphabets.)
Finally, we also prove several useful continuity properties of quantum information quantities. These allow us to establish that the third-order term in the Strassen-type asymptotic expansion in \eqref{eqn:strass} for c-q channels with discrete support is $O(\log n)$, as in the classical case.


\section{Preliminaries}
\label{sec:pre}


We consider the real vector space of self-adjoint (Hermitian) operators on a finite-dimensional inner product (Hilbert) space. We denote the space of self-adjoint operators by $\cH$ and keep it fixed throughout to ease notation.
For $A, B \in \cH$, we write $A \geq B$ iff $A - B$ is positive semi-definite.
Moreover, we denote by $\{ A > B \}$ and $\{ A \geq B\}$ 
the projectors onto the positive and non-negative subspaces of $A - B$, respectively. We write $A \gg B$ to denote the fact that the kernel of $A$ is contained in the kernel of $B$. Let $\lambda_{\min}(A)$ denote the minimum eigenvalue of $A$.
We equip $\cH$ with a metric, the trace distance $\delta_{\tr}(A,B) := \frac12 \tr | A - B |$, where $\tr$ denotes the trace.
The \emph{identity operator} is denoted by $\id$.
The set of \emph{quantum states} is given by $\cS := \{ \rho \in \cH \,|\, \rho \geq 0 \land \tr(\rho) = 1 \}$. 
Clearly, $(\cS, \delta_{\tr})$ is a compact metric space.

For any closed (and thus compact) subset $\cSo \subseteq \cS$, we denote by $\cP(\cSo)$ the set of probability measures on $(\cSo, \Sigma_{\circ})$, where $\Sigma_{\circ}$ is the Borel $\sigma$-algebra on $(\cSo, \delta_{\tr})$. Since $(\cSo, \delta_{\tr})$ is a compact metric space, $\big(\cP(\cSo), \delta_{\textrm{wc}}\big)$ is a compact metric space, where $\delta_{\textrm{wc}}$ denotes the Prohorov metric~\cite[Sec.~6 and Thm.~6.4]{partha67}. We will not use $\delta_{\textrm{wc}}$ explicitly but simply note that convergence in $\delta_{\textrm{wc}}$ is equivalent to weak convergence of probability measures. As such, any function of the form
\begin{align}
  \cP(\cSo) \to \mathbb{R}, \qquad \bbP \mapsto \int_{\cSo} \diff \bbP(\rho) f(\rho) 
\end{align}
is continuous if $f$ is bounded and continuous. 
If $\cSo$ is discrete, we abuse notation and also use $\cP(\cSo)$ to denote the set of probability mass functions on $\cSo$. We then use $P \in \cP(\cSo)$ to denote its elements. We often use the abbreviations $\rho^{(\bbP)}$ and $\rho^{(P)}$ to denote the averaged states
\begin{align}
\rho^{(\bbP)} := \int_{\cSo} \diff \bbP(\rho)\rho \qquad \textrm{and} \qquad 
\rho^{(P)} := \sum_{\rho \in \cSo} P(\rho)\rho .
\end{align}

For any $n \in \mathbb{N}$, we also consider the $n$-fold products of the underlying inner-product space and denote the associated set of self-adjoint operators and states with $\cH^n$ and $\cS^n$, respectively. For any $\cSo \subseteq \cS$, we denote by $\cSo^{\otimes n} \subseteq \cS^n$ the set of $n$-tuples of states in $\cSo$, represented as a product state $\bigotimes_{i=1}^n \rho_i$, where $\rho_i \in \cSo$. Clearly, $\cS^{\otimes n} \subseteq \cS^n$.

We employ the cumulative distribution function of the standard normal distribution 
\begin{align}
\Phi(a) := \int_{-\infty}^a \frac{1}{\sqrt{2 \pi}} \exp \Big(-\frac12 x^2\Big) \,\mathrm{d}x
\end{align}
and define its inverse as $\Phi^{-1}(\eps) := \sup \{ a \in \mathbb{R} \,|\, \Phi(a) \leq \eps \}$, which reduces to the usual inverse for $0 < \eps < 1$ and extends to take values $\pm \infty$ outside that range.

\subsection{Codes for Classical-Quantum Channels}

We consider general c-q channels, i.e.\ arbitrary functions $\cW: \cX \to \cS$, where $\cX$ is an arbitrary set. A special case of this is a quantum channel, namely a completely positive trace-preserving (CPTP) map $\cW: \cS' \to \cS$, where $\cS'$ denotes a set of quantum states.
We denote the image of the channel by
  \begin{align}
    \im(\cW) := \{ \rho \in \cS \,|\, \exists\, x \in \cX : \rho = \cW(x) \} ,
  \end{align}
and its closure by $\overline{\im(\cW)}$. Without loss of generality, we may assume that $\im(\cW)$ has full support on the underlying Hilbert space, i.e.\ every vector (of the underlying Hilbert space) is supported by at least one element in $\im(\cW)$. Thus, we will usually set $d = |\supp(\im(\cW))|$. 

A {\em code} $\cC$ for $\cW$ is defined by the triple $\{\cM, e, \cD\}$,
where $\cM$ is a (discrete) set of messages, $e: \cM \to \cX$ an encoding function and $\cD = \{ Q_m \}_{m \in \cM}$ is a positive operator valued measure (POVM).\footnote{A POVM in this context is a set of operators $\{ Q_m \}_{m \in \cM}$ satisfying $Q_m \geq 0$ for all $m \in \cM$ and $\sum_{m \in \cM} Q_m = \id$.}
We write $\abs{\cC} = \abs{\cM}$ for the cardinality of the message set. We define the \emph{average error probability} of a code $\cC$ for the channel $\cW$  as 
\begin{align}
\perr(\cC,\cW):=  1-\frac{1}{ |\cM|}\sum_{m\in\cM} \tr \big( \cW(e(m)) Q_{m} \big)
\end{align} 
where the distribution over messages $M$ is assumed to be uniform on $\cM$. Alternatively, we may write
$\perr(\cC,\cW) = \Pr[M \neq M']$ where
\begin{align}
  M \xrightarrow{\ e\ } X \xrightarrow{\ \cW\ } \cW(X) \xrightarrow{\ \cD\ } M'
\end{align}
forms a Markov chain, $\cW(X)$ denotes the (random) output of the channel, and $M'$ thus denotes the output of the decoder. 

To characterize the non-asymptotic fundamental limit of data transmission over a single use of the channel, we define the \emph{maximum size of a codebook for $\cW$ with average error $\eps$} as
\begin{align}
  M^*(\cW,\eps) := \max \big\{ m \in \mathbb{N} \,\big|\, \exists\, \cC :\ \abs{\cC} = m \ \land\ \perr(\cC, \cW) \leq \eps \big\} .
\end{align}
We are interested to evaluate this quantity for the composite channel $\cW^n$, corresponding to $n \geq 1$ uses of a memoryless channel $\cW$. 
Formally, the $n$-fold i.i.d.\ repetition of the channel, $\cW^n: \cX^n\to \cS^{\otimes n}$, takes as input a vector $\vec{x} =( x_1,x_2, \ldots , x_n) \in \cX^n$ and maps it to $\cW(x_1) \otimes \cW(x_2) \otimes \ldots \otimes \cW(x_n) \in \cS^{\otimes n}$. 
In particular, this model does not allow for entangled channel outputs. The non-asymptotic fundamental limit of data transmission over $n$ uses of the channel is consequently given by $M^*(\cW^n,\eps)$.

\subsection{Information Quantities}

The following basic quantities are of interest here. For any $\rho \in \cS$, we employ the von Neumann entropy
$H(\rho) := - \tr (\rho \log \rho)$.
Moreover, for positive semi-definite $\sigma$ satisfying $\sigma \gg \rho$, the \emph{relative entropy}~\cite{umegaki62,hiai91} and the \emph{relative entropy variance}~\cite{tomamichel12,li12} are respectively defined as
\begin{align}
  D(\rho\|\sigma) &:= \tr \Big( \rho \big( \log \rho - \log \sigma \big) \Big) \qquad \textrm{and}\\
  V(\rho\|\sigma) &:= \tr \Big( \rho \big( \log \rho - \log \sigma - D(\rho\|\sigma)\cdot \id \big)^2 \Big) \,.
\end{align}
As usual, we implicitly use the convention $0 \log^k 0 \equiv 0$ for all $k \in \mathbb{N}$.

Classically, for two probability mass functions $P, Q \in \cP(\cX)$, the 
relative entropy
$D(P\|Q)$ is the expectation value of the log-likelihood ratio $\log \big( P(X)/Q(X)\big)$ where $X \leftarrow P$, and $V(P\|Q)$ is the corresponding variance. The above definition of $V(\rho\|\sigma)$ is thus a natural non-commutative generalization of the classical concept, with its operational meaning firmly established in~\cite{tomamichel12,li12}.

We summarize some properties of the above quantities, which we will employ later.
\begin{enumerate}[itemsep=0mm]
  \item $\rho \mapsto H(\rho)$ is strictly concave (cf., e.g., Lemma~\ref{lm:strict}) and continuous.
  \item $(\rho,\sigma) \mapsto D(\rho\|\sigma)$ is jointly convex and lower semi-continuous. In fact, it is continuous except when it diverges to infinity, i.e.\ when $\sigma \not\gg \rho$.
  \item $D(\rho\|\sigma)$ is positive definite, i.e.\ $D(\rho\|\sigma) \geq 0$ with equality iff $\rho = \sigma$.
  \item $(\rho,\sigma) \mapsto V(\rho\|\sigma)$ is continuous except when $\sigma \not\gg \rho$.
\end{enumerate}

Finally, in order to express the one-shot bounds, we introduce the \emph{$\eps$-hypothesis-testing divergence}~\cite{wang10}. 
For any $\eps \in (0,1)$ and $\rho, \sigma \in \cS$, it is defined as
\begin{align}
  D_h^{\eps}(\rho\|\sigma) := - \log \frac{\beta_{1-\eps}(\rho\|\sigma)}{1-\eps}, \qquad \textrm{where} \quad \beta_{1-\eps}(\rho\|\sigma) := \min_{0 \leq Q \leq \id \atop \tr(Q \rho) \geq 1-\eps}  \tr(Q \sigma) \,. \label{eq:defhypo}
\end{align} 
Note that $\beta_{1-\eps}$ is the smallest type-II error of a hypothesis
test between $\rho$ and $\sigma$ with type-I error at most $\eps$.
The $\eps$-hypothesis testing divergence satisfies the following basic properties, which we summarize here for later reference.
\begin{lemma} \label{lm:hypo-prop}
  Let $\eps \in (0,1)$, let $\cSo,\cSo' \subseteq \cS$ be discrete sets, and let $P \in \cP(\cSo)$, $Q \in\cP(\cSo')$. Define $\rho = \sum_{\tau \in \cSo} P(\tau) \tau$ and $\sigma = \sum_{\omega \in \cSo'} Q(\omega) \omega$. 
  Then $D_h^{\eps}(\rho\|\sigma)$ satisfies the following properties:
  \begin{enumerate}[itemsep=0mm]
  \item $D_h^{\eps}(\rho\|\sigma) \geq 0$ with equality if and only if $\rho = \sigma$. \emph{(cf.~\cite[Prop.~3.2]{dupuis12})}
  \item For any CPTP map $\mathcal{M}$ we have $D_h^{\eps}(\rho\|\sigma) \geq D_h^{\eps}\big(\mathcal{M}(\rho) \big\| \mathcal{M}(\sigma) \big)$. \emph{(cf.~\cite{wang10})}
   \item $D_h^{\eps}(\rho\|\sigma) \leq  \min_{\omega \in \cSo'} \big\{ D_h^{\eps}(\rho\|\omega) + \log \frac{1}{Q(\omega)} \big\}$.
   \item $D_h^{\eps}(\rho\|\sigma) \leq \max_{\tau \in \cSo} D_h^{\eps}(\tau\|\sigma)$.
  \end{enumerate}
\end{lemma}
\noindent The last inequality shows that $\rho \mapsto D_h^{\eps}(\rho \| \sigma)$ is quasi-convex. The last two inequalities can be verified by a close inspection of the definition in~\eqref{eq:defhypo} and we omit the proof.

\section{Main Results}\label{sec:results}

\subsection{The Divergence Radius of a Set of Quantum States}
\label{sec:res-radius}

It is well known that the capacity of a classical or classical-quantum channel can be represented geometrically as the divergence radius of the channel image. (For the quantum case, see, e.g.~\cite{ohya97} and~\cite{schumacher01}.) Here, we take a complementary approach and investigate the divergence radius of subsets of the set of quantum states. If such a set is the image of a channel, our analysis allows us to construct capacity-achieving ensembles by just looking at the channel image. Furthermore, this viewpoint leads to a natural quantum generalization of the concept of channel dispersion. Thus, somewhat surprisingly, we will see that not only the capacity but also the finite blocklength behavior of channels is governed by the geometry of the channel image.

\subsubsection{Divergence Radius}

Let us start by investigating the divergence radius of arbitrary closed subsets of the set of quantum states on a finite-dimensional Hilbert space.

\begin{definition}
  \label{def:radius}
Let $\cSo \subseteq \cS$ be closed. The \emph{divergence radius} of $\cSo$ (in $\cS$) is defined as
\begin{align}
  \chi(\cSo) := \inf_{\sigma \in \cS} \sup_{\rho \in \cSo} D(\rho\|\sigma) \,. \label{eq:dr}
\end{align}
\end{definition}

We show the following properties of the divergence radius.

\begin{theorem}
\label{th:radius}
  Let $\cSo \subseteq \cS$ be closed. We find the following:
  \begin{enumerate}[itemsep=0mm]
    \item The \emph{divergence center}, defined as $\sigma^*(\cSo) := \argmin_{\sigma \in \cS} \big\{ \sup_{\rho \in \cSo} D(\rho\|\sigma) \big\}$, exists and is unique. Moreover, $\sigma^*(\cSo) \gg \rho$ for all $\rho \in \cSo$.
    \item Define the set of \emph{peripheral points} of $\cSo$, i.e.\
    \begin{align}
      \label{eq:gamma}
      \Gamma(\cSo) := \argmax_{\rho \in \cSo} D\big(\rho \big\|\sigma^*(\cSo) \big) .
    \end{align}
     Then, $D\big(\rho \big\|\sigma^*(\cSo) \big) \leq \chi(\cSo)$ for all $\rho \in \cSo$ with equality iff $\rho \in \Gamma(\cSo)$.    
    \item We have $\sigma^*(\cSo) \in \conv(\Gamma(\cSo))$.
    \item The divergence radius has the following alternative representation:
    \begin{align} \label{eq:holevo}
        \chi(\cSo) = \sup_{\bbP \in \cP(\cSo)} \int \diff \bbP(\rho)\, D\bigg(\rho\,\bigg\|\, \int \diff \bbP(\rho') \rho' \bigg) .
    \end{align}
    \item
The set of probability measures that achieve the supremum is given
    by the \emph{peripheral decompositions of the divergence center}, namely the compact convex
    set
    \begin{align}
      \Pi(\cSo) := \bigg\{ \bbP \in \cP\big(\Gamma(\cSo)\big) \,\bigg|\,  \int \diff \bbP(\rho) \rho = \sigma^*(\cSo) \bigg\} . \label{eq:pi}
    \end{align}
   Moreover, $\Pi(\cSo)$ contains a discrete probability measure with support on at most $d^2$ points in $\Gamma(\cSo)$. 
       \end{enumerate}
\end{theorem}

The proof of this theorem is presented in Section~\ref{sec:radius} and we illustrate it in Figure~\ref{fig:radius}.

\begin{remark}
  Uniqueness of $\sigma^*(\cSo)$ was also claimed by Ohya, Petz and Watanabe~\cite[Lem.~3.4]{ohya97} in a related context. However, they argue that this directly follows from the ``fact that the relative entropy functional is strictly convex in the second variable''. We submit that more care has to be taken to establish uniqueness. Notably, the functional $\sigma \mapsto D(\rho\|\sigma)$ is only strictly convex if $\rho > 0$ is positive definite and trivial counterexamples can be constructed otherwise. It is then unclear how to apply this property directly to the situation at hand.
\end{remark}

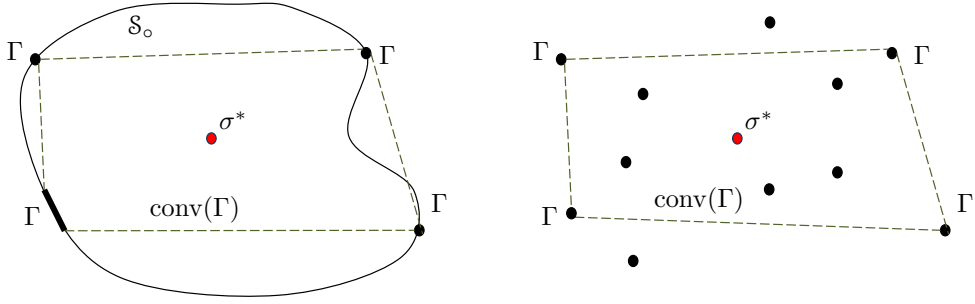
\begin{figure}[t]
\begin{center}
\begin{overpic}[width=.35\columnwidth]{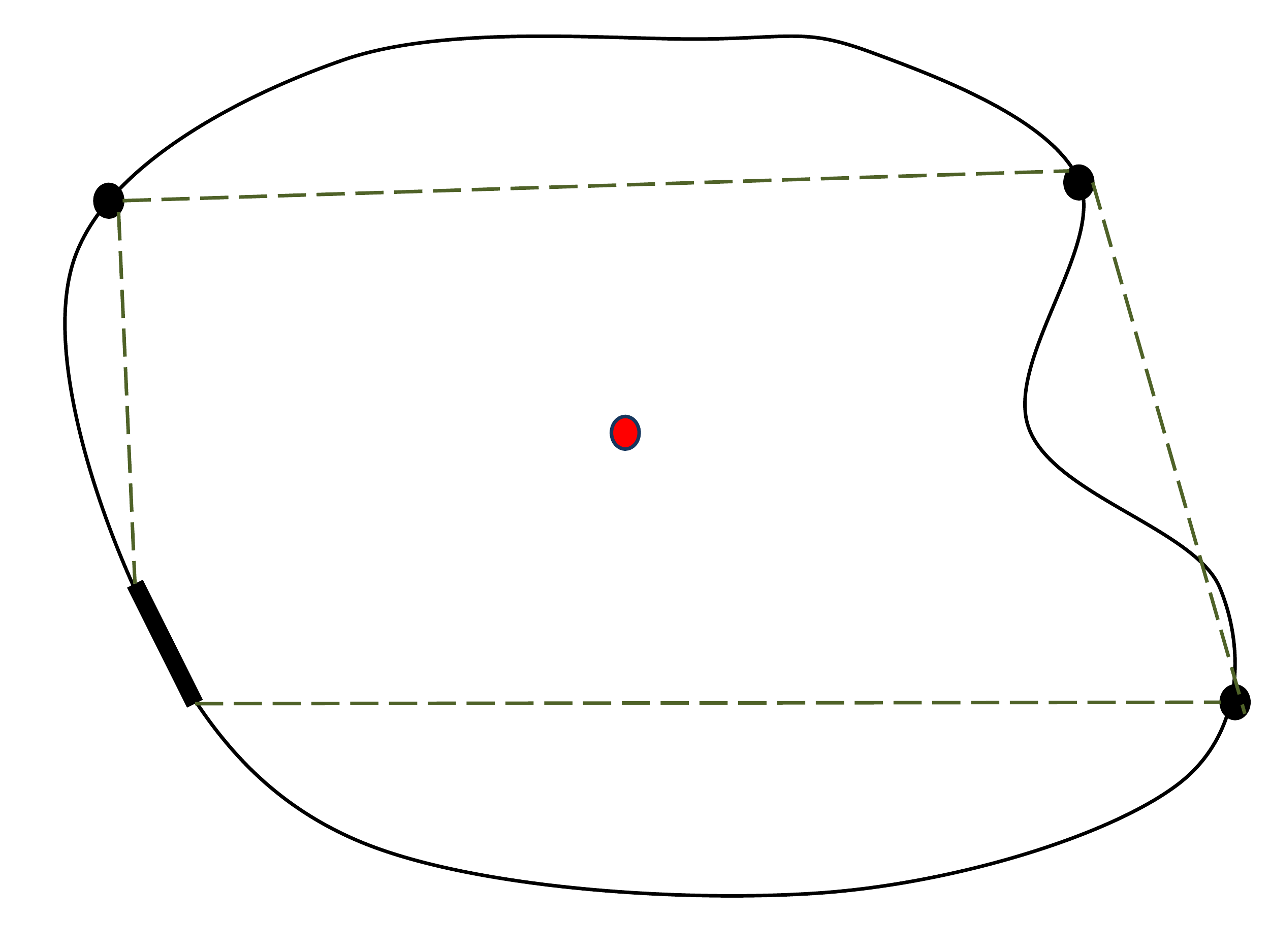}
\put(51,43){$\sigma^*$}
\put(30,65){$\cSo$}
\put(100,24){$\Gamma$}
\put(88,58){$\Gamma$}
\put(2,59){$\Gamma$}
\put(6,21){$\Gamma$}
\put(35,23){$\conv(\Gamma)$}
\end{overpic}
\hspace{1cm}
\begin{overpic}[width=.35\columnwidth]{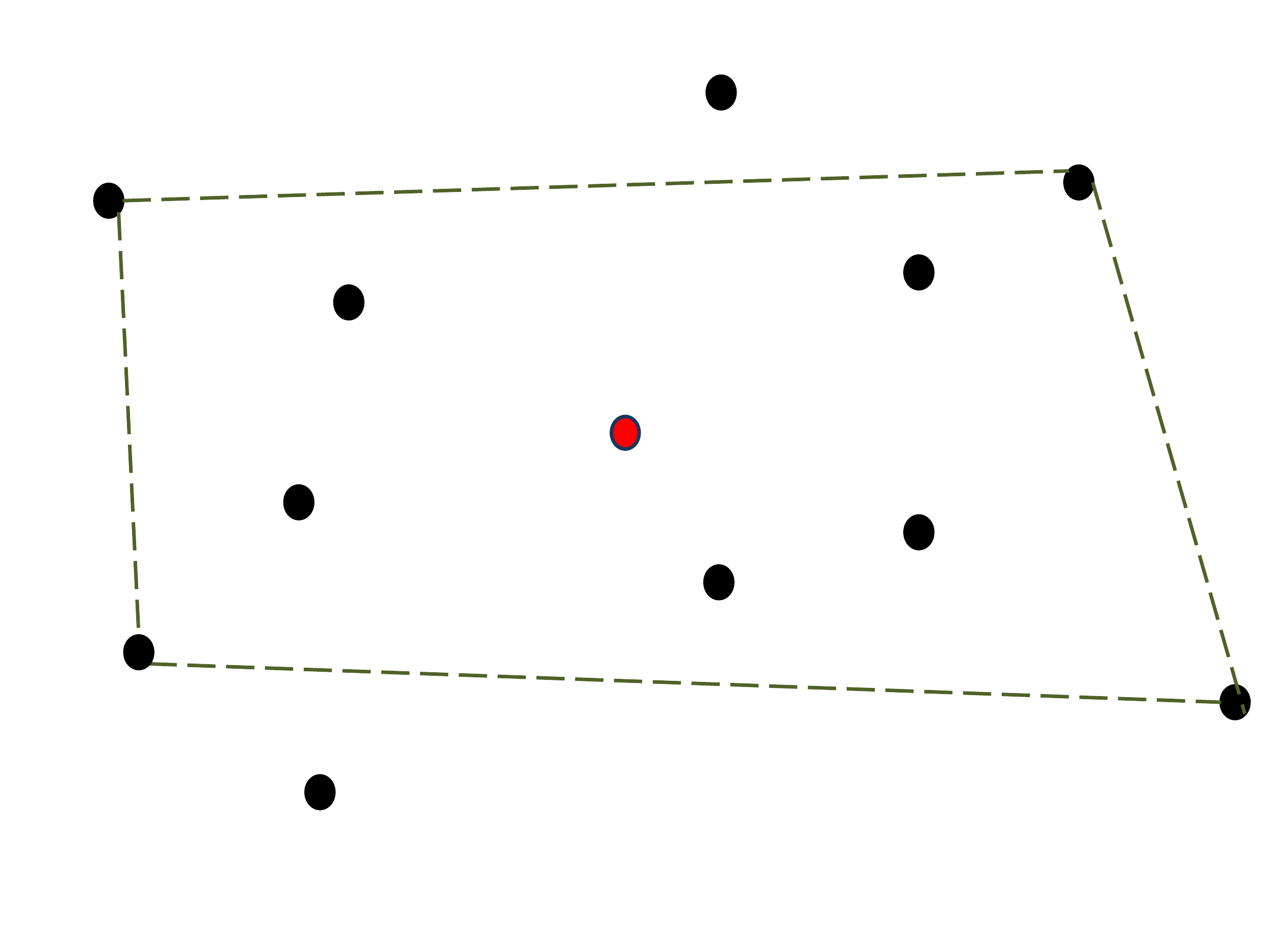}
\put(51,43){$\sigma^*$}
\put(100,24){$\Gamma$}
\put(90,58){$\Gamma$}
\put(2,59){$\Gamma$}
\put(4,21){$\Gamma$}
\put(30,25){$\conv(\Gamma)$}
\end{overpic}
\end{center}
\vspace{-0.5cm}

\caption{Example of a discrete and a continuous set with divergence center $\sigma^*$ and peripheral states $\Gamma$. The set $\Gamma$ must lie on the boundary of $\cSo$ due to the quasi-convexity of $\rho \mapsto D(\rho\|\sigma)$ (cf.~Lemma~\ref{lm:hypo-prop}). As seen in Theorem~\ref{th:radius}, the center $\sigma^*$ lies in the convex hull of $\Gamma(\cSo)$ consistent with the Euclidean intuition. }
\label{fig:radius}
\end{figure}

\begin{remark}
Property 3 is of particular importance for our argument and has not been shown before. A weaker property, namely $\sigma^*(\cSo) \in \conv(\cSo)$ was already pointed out in~\cite[Lem.~3.4]{ohya97}.
 However, our stronger Property 3 implies that $\sigma^*(\cSo)$ can be written as a convex combination of states in $\Gamma(\cSo)$, i.e.\ $\sigma^*(\cSo) = \rho^{(\bbP)}$ for some $\bbP \in \cP(\Gamma(\cSo))$. 
If $\cSo$ is the image of a quantum channel $\cW$, we write $\cW^{-1}(\Gamma(\cSo))$ to denote any pre-image of $\Gamma(\cSo)$. Then, the tuple $\{\bbP, \cW^{-1}(\Gamma(\cSo))\}$ corresponds to an optimal ensemble of input states, i.e.\ an ensemble that achieves the maximum Holevo information. 
 In particular, $\Pi(\cSo)$ as defined in~\eqref{eq:pi} is non-empty. 
\end{remark}

\begin{remark}
It is natural to see~\eqref{eq:holevo} as the dual problem (cf.~\cite{boyd04}) to the convex optimization problem in~\eqref{eq:dr}; in particular, the integral in~\eqref{eq:holevo} is concave in $\bbP$. As such \eqref{eq:holevo} implies strong duality.\footnote{A dual problem to~\eqref{eq:holevo} for the discrete case has also been established in~\cite{sutter14}, but elementary manipulations reveal that the dual program there is equivalent to the divergence radius optimization in~\eqref{eq:dr}.}
\end{remark}

\subsubsection{Peripheral Information Variance}

The above observations allow us to define the minimal and maximal peripheral information variance of $\cSo$ in terms of the information variance of peripheral decompositions of the divergence center.  To do so, we consider measures $\bbP \in \Pi(\cSo)$ and optimize 
\begin{align}
  V(\bbP|\sigma^*(\cSo)), \qquad \textrm{where} \quad V(\bbP|\sigma) := \int \diff \bbP(\rho)\, V( \rho \big\| \sigma) \,.
\end{align}
is the \emph{conditional information variance}.
This leads to the following definitions.

\begin{definition}
  Let $\cSo \subseteq \cS$ be closed and $\Pi(\cSo)$ defined in~\eqref{eq:pi}. Then, the \emph{minimal and maximal peripheral information variance} of $\cSo$ (in $\cS$) are respectively defined as
  \begin{align}
    v_{\min}(\cSo) &:= \inf_{\bbP \in \Pi(\cSo)} V\big(\bbP \big|\sigma^*(\cSo)\big) = \inf_{\bbP \in \Pi(\cSo)} \int \diff \bbP(\rho)\, V\big( \rho \big\| \sigma^*(\cSo) \big), \qquad \textrm{and} \label{eq:vmin}\\
    v_{\max}(\cSo) &:= \sup_{\bbP \in \Pi(\cSo)} V\big(\bbP \big|\sigma^*(\cSo)\big)= \sup_{\bbP \in \Pi(\cSo)}\int \diff \bbP(\rho)\, V\big( \rho \big\| \sigma^*(\cSo) \big) .\label{eq:vmax}
  \end{align}
\end{definition}

It is evident from the compactness of $\Pi(\cSo)$ that the infimum and supremum are achieved so we may replace $\inf$ and $\sup$ with $\min$ and $\max$, respectively. Moreover, the minimum in Eq.~\eqref{eq:vmin} is achieved for a probability measure $\bbP \in \cP(\Gamma(\cSo))$ that satisfies the linear constraints
\begin{align}
  &\int \diff \bbP(\rho) \rho = \sigma^*(\cSo) \quad \textrm{and} \quad \int \diff \bbP(\rho)\,V\big(\rho\big\|\sigma^*(\cSo)\big) = v_{\min}(\cSo).
\end{align}
These constitute $d^2-1$ real constraints for the first equality and one additional constraint for the second one. Since $\Gamma(\cSo)$ is not connected in general, Caratheodory's theorem (see, e.g., \cite[Thm.~18]{eggleston58}) yields the following lemma:

\begin{lemma}
  \label{lm:cara}
  There exist discrete probability measures with support on at most $d^2+1$ points in $\Gamma(\cSo)$ that achieve the infimum and supremum in~\eqref{eq:vmin} and~\eqref{eq:vmax}, respectively.   
\end{lemma}

\subsection{Second-Order Approximation for the Classical Capacity}
\label{sec:res-main}

\subsubsection{Capacity of Classical-Quantum Channels}

Our main result is the evaluation of the second-order asymptotics for the capacity of c-q channels with general input.
(Recall that we consider general channels $\cW: \cX \to \cS$, where $\cX$ is an arbitrary set\footnote{In particular, this set is not assumed to be countable or have any topological structure.} and $\cS$ is the set of quantum states on an arbitrary finite-dimensional Hilbert space.)

\begin{theorem}
  \label{th:main}
  Let $\eps \in (0,1)$ and $\cW$ be a c-q channel. Setting $\cSo = \overline{\im(\cW)}$, we find
  \begin{align}
  &\log M^*(\cW^n, \eps) = n\, C(\cW) + \sqrt{n\, V_{\eps}(\cW)}\, \Phi^{-1}(\eps) + K(n, \cSo, \eps), \quad \textrm{where} \\
   &\qquad \quad C(\cW) = \chi(\cSo) \quad \textrm{and} \quad V_{\eps}(\cW) = v_{\eps}(\cSo) := \begin{cases} v_{\min}(\cSo) & \textrm{if }\ 0 < \eps \leq \frac12 \\
       v_{\max}(\cSo) & \textrm{if }\ \frac12 < \eps < 1 \end{cases} .
  \end{align}
  We have $K(n, \cSo, \eps) = o(\sqrt{n})$ for all channels. Moreover, if $\cSo$ is finite and $v_{\eps}(\cSo) > 0$, we have $K(n, \cSo, \eps) = O(\log n)$. 
\end{theorem}

\begin{remark}
  The \emph{$\eps$-channel dispersion} is an operational quantity defined as~\cite[Eq. (221)]{polyanskiy10}
  \begin{align}
     V_{\eps}(\cW) := \limsup_{n\to\infty} \frac{1}{n} \bigg( \frac{n C(\cW) - \log M^*(\cW^n, \eps)}{\Phi^{-1}(\eps)} \bigg)^2 .
  \end{align}
  Our results imply that it equals $v_{\eps}(\cSo)$, the minimal or maximal peripheral information variance of the channel image, depending on the value of $\eps$.
\end{remark}

\begin{remark}
  Traditionally, classical-quantum channels are studied for the case when $\cX$ is discrete. In our framework, this corresponds to a discrete set $\cSo = \{ \cW(x) \,|\, x \in \cX \}$.
\end{remark}

\begin{remark}
Some restrictions on $\cSo$ are necessary in order to show that $K(n, \cSo,\eps) = O(\log n)$. Indeed, there exists a class of classical discrete memoryless channels, so-called exotic channels~\cite[p.~2231 and App.~H]{polyanskiy10}, for which $v_{\eps}(\cSo) = 0$ and $K(n, \cSo, \eps) = \Theta(n^{1/3})$ hold~\cite[Thm.~51]{polyanskiythesis10}.
\end{remark}

\noindent We sketch the main ideas and outline of our proof in the following.

\paragraph{Summary of the Proof of the Direct Part:}
The direct part of Theorem~\ref{th:main}, established in Section~\ref{sec:direct}, is derived employing a one-shot bound due to Wang and Renner that relates $M^*(\cW^n, \eps)$ with the $\eps$-hypothesis-testing divergence, $D_h^{\eps}(\cdot\|\cdot)$, defined in~\eqref{eq:defhypo} above. The bound is valid for classical-quantum channels with finite input alphabets and the asymptotics are derived in this setting based upon the second-order asymptotics of the hypothesis testing divergence evaluated on i.i.d.\ states established in~\cite{li12} and~\cite{tomamichel12}. Finally, a simple application of Caratheodory's theorem (Lemma~\ref{lm:cara}) shows that it is possible to achieve the second-order asymptotics with finite alphabets (of size depending on the dimension of the output space).

\paragraph{Summary of the Proof of the Converse Part:}
The converse part of Theorem~\ref{th:main} is proved in Sections~\ref{sec:proof/one-shot-converse}--\ref{sec:converse}. The proof employes a one-shot analogue of the divergence radius in Definition~\ref{def:radius}.
\begin{definition}
Let $\eps \in (0,1)$ and $\cSo \subseteq \cS$. The $\eps$-\emph{hypothesis-testing divergence radius} is defined as
\begin{align}
  \chi_h^{\eps}(\cSo) := \inf_{\sigma \in \cS} \sup_{\rho \in \cSo} D_h^{\eps}(\rho\|\sigma) .
\end{align}
\end{definition}

This quantity, evaluated for the channel image, constitutes an upper bound on $M^*(\cW, \eps)$ for c-q channels with general input. In Section~\ref{sec:proof/one-shot-converse}, we establish the following one-shot converse bound:

\begin{proposition} \label{pr:one-shot-converse}
  Let $\eps \in (0,1)$ and let $\cW$ be a c-q channel. For any $\mu \in (0, 1-\eps)$, we have
  \begin{align}
  \log M^*(\cW,\eps ) \leq \chi_h^{\eps+\mu}\Big(\overline{\im(\cW)}\Big) + \log \frac{\eps+\mu}{\mu(1-\eps-\mu)} . \label{eq:conv}
  \end{align}
\end{proposition}

This bound should be compared to the bounds by Renner-Wang~\cite{wang10} and Matthews-Wehner~\cite{matthews12}. Both of these works also establish one-shot converse bounds in terms of the $\eps$-hypothesis testing divergence (see also~\cite[Remark~15]{hayashi03}). However, our result crucially differs in that our bound only depends on the image of the channel, independently of the input alphabet supported by the channel. It thus allows us to treat the remaining evaluation as a problem on the output space.

Applied to the $n$-fold memoryless repetition of the c-q channel~$\cW$, it yields
\begin{align}
  \log M^*(\cW^n,\eps ) \leq \chi_h^{\eps+\mu}(\cSo^{\otimes n}) + O(\log n) . \label{eq:sketch1}
\end{align}
where $\mu$ is chosen inversely polynomial in $n$. Proposition~\ref{pr:thedifficultpart} in Section~\ref{sec:converse}, then establishes that
\begin{align}
  \chi_h^{\eps+\mu}\big(\cSo^{\otimes n}\big) \leq n\, \chi(\cSo) + \sqrt{n\, v_{\eps}(\cSo)}\, \Phi^{-1}(\eps) + o(\sqrt{n}) \label{eq:sketch2}\,,
\end{align}
which, combined with~\eqref{eq:sketch1}, concludes the proof.

This asymptotic expansion in~\eqref{eq:sketch2} constitutes the technically most challenging part of our derivation.
To evaluate these asymptotics for a suitable choice of $\sigma^n$ we extend the second-order approximation of~\cite{tomamichel12} to non-identical product distributions. Moreover, we show that these bounds hold uniformly in all sequences $\rho^n = \bigotimes_{i=1}^n \rho_i \in \cSo^{\otimes n}$ that appear in the supremum above. This is particularly challenging because we have to treat separately sequences for which the average relative entropy variance is small, and hence the convergence to the second-order approximation is too slow.\footnote{For a classical analogue, recall that the convergence speed in the Berry-Esseen theorem is inversely proportional to $\sigma^{3}$, where $\sigma^2$ is the average variance of a sequence of non-i.i.d.~random variables.}
To tackle this, we employ a net on $\cSo$ and in particular do not appeal to the use of constant composition codes and type-counting arguments, which are workhorses 
of the second-order analysis for discrete memoryless
channels in the classical setting. Our novel proof thus departs from the usual treatment, which in particular allows us to consider general input alphabets.

\subsubsection{Classical Capacity for Image-Additive Quantum Channels}
\label{sec:image-add}

First, note that the achievability bounds in Theorem~\ref{th:main} in fact apply for the classical capacity of all quantum channels, and can be achieved using product states. To see this, let $\cX$ be a set of quantum states (whether the states in $\cX$ are modeled as density operators on a Hilbert space or states of a C* algebra is irrelevant here)
and $\cW$ be the quantum channel from $\cX$ to $\cS$, as usual. Obviously the channel is now a completely positive trace-preserving map, but we do not need to use this structure here and focus again on its image, $\cSo = \im(\cW)$, where closure is now unnecessary since the image is compact.
Thus, for all quantum channels $\cW$, we have\footnote{But note that $\chi(\cSo)$ could generally be smaller than $C(\cW)$.}
  \begin{align}
  &\log M^*(\cW^n, \eps) \geq n\, \chi(\cSo) + \sqrt{n\, v_{\eps}(\cSo)}\, \Phi^{-1}(\eps) + o(\sqrt{n}) \,.
\end{align}

Moreover, the converse part of the proof of Theorem~\ref{th:main} can be easily adapted to cover general image-additive quantum channels. 
The logarithm of the maximum codebook size of a quantum channel is certainly also upper bounded by $\chi_h^{\eps}(\cSo)$ as in~\eqref{eq:sketch1}, so in particular we find
\begin{align}
    \log M^*(\cW^n,\eps ) \leq \chi_h^{\eps+\mu}(\cSo^n) + \log \frac{\eps+\mu}{\mu(1-\eps-\mu)} , \qquad \textrm{where} \quad \cSo^n = {\im(\cW^n)} \,.
\end{align}

However, the crucial difference vis-\`a-vis classical-quantum channels is that here we generally have $\cSo^n \neq \cSo^{\otimes n}$ as the channel image can be enlarged in the presence of non-product input states. Restricting to image-additive channels $\cW$, however, we find
\begin{align}
  \cSo^n = {\im(\cW^n)} = { \conv(\im(\cW)^{\otimes n}) } = \conv(\cSo^{\otimes n})
\end{align}
Now the only missing observation is that $\chi_h^{\eps}(\conv(\cSo^{\otimes n})) = \chi_h^{\eps}(\cSo^{\otimes n})$ for all $\cSo \subseteq \cS$, which is an immediate consequence of the quasi-convexity of $\rho \mapsto D_h^{\eps}(\rho\|\sigma)$, shown in Part~4 of Lemma~\ref{lm:hypo-prop}.
%
Hence, Proposition~\ref{pr:thedifficultpart} directly applies to this situation as well and we arrive at the following corollary:

\begin{corollary}
  Let $\eps \in (0,1)$ and $\cW$ be an image-additive quantum channel. Then,
  \begin{align}
  &\log M^*(\cW^n, \eps) = n\, C(\cW) + \sqrt{n\, V_{\eps}(\cW)}\, \Phi^{-1}(\eps) + o(\sqrt{n})  \end{align}
  with $C(\cW)$ and $V_{\eps}(\cW)$ as defined in Theorem~\ref{th:main}.
\end{corollary}

Similarly, if the input of the channel is restricted to separable states then clearly 
the restricted image satisfies $\im_{\textrm{sep}}(\cW^n) = \conv(\im(\cW)^{\otimes n})$ and thus Proposition~\ref{pr:thedifficultpart} again suffices to determine the second-order asymptotics.\footnote{The first-order asymptotics (for the case of product state inputs) were discussed in detail in~\cite{fujiwara98}.}

\begin{corollary}
  \label{cor:qq} Let $\eps \in (0,1)$, let $\cW$ be any quantum channel. Let $M_{\textrm{sep}}^*(\cW^n, \eps)$ denote the maximum size of a codebook for classical information transmission over $\cW$ with average error $\eps$ when the channel is restricted to separable input states. Then,
  \begin{align}
  \log M_{\textrm{sep}}^*(\cW^n, \eps) =  n\, C(\cW) + \sqrt{n\, V_{\eps}(\cW)}\, \Phi^{-1}(\eps) + o(\sqrt{n}) , \label{eq:cor}
  \end{align}
    with $C(\cW)$ and $V_{\eps}(\cW)$ as defined in Theorem~\ref{th:main}.
\end{corollary}

\begin{example}
Qubit Pauli channels are symmetric under reflection at the center of the Bloch sphere. As such, $\sigma^*(\cSo) = \frac12 \id$ and it is furthermore easy to verify that any capacity-achieving ensemble (of minimal size) is commutative. Hence, the capacity and dispersion of a Pauli channel equal those of a (classical) binary symmetric channel (see, e.g.,~\cite[Thm.~52]{polyanskiy10}).
\end{example}

\begin{figure}[t]
  \subfigure[\,The sets $\cSo^{\gamma} = \im(\cE_{\textrm{ad}}^{\gamma})$ projected onto the xz-plane of the Bloch sphere for $\gamma \in \{0, \frac14, \frac34\}$.\label{fig:ad1}]{
\begin{overpic}[width = .35\columnwidth]{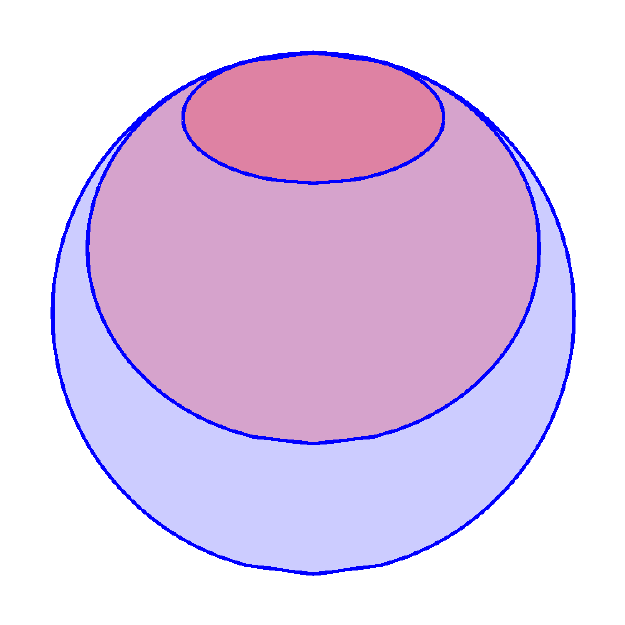}
  \put(50,91.8){\circle*{2}}
  \put(48,95){$|0\rangle$}
  \put(50,8.8){\circle*{2}}
  \put(48,2){$|1\rangle$}
  \put(8.5,50){\circle*{2}}
  \put(-2,50){$|+\rangle$}
  \put(91.8,50){\circle*{2}}
  \put(94,50){$|-\rangle$}
  
  \put(62,15){$\cSo^0$}
  \put(50,66){\circle*{2}}
  \put(45,58){$\sigma^*(\cSo^{\frac14})$}
  \put(65,37){$\cSo^{\frac14}$}
  \put(14.5,66){\circle*{2}}
  \put(85.3,66){\circle*{2}}
  \multiput(14.5,66)(5,0){14}{\line(1,0){3}}
  
  \put(16,58){$\Gamma(\cSo^{\frac14})$}
  \put(70,58){$\Gamma(\cSo^{\frac14})$}

  \put(50,85){\circle*{2}}
  \put(42,77){$\sigma^*(\cSo^{\frac34})$}
  \put(62,77){$\cSo^{\frac34}$}  
  \put(30.5,85){\circle*{2}}
  \put(69.5,85){\circle*{2}}
  \multiput(30.5,85)(5,0){8}{\line(1,0){3}}

  \put(16,89){$\Gamma(\cSo^{\frac34})$}
  \put(71,89){$\Gamma(\cSo^{\frac34})$}
\end{overpic}}
  \hspace{10cm}
  \subfigure[\,Divergence radius, $\chi(\cSo^{\gamma})$ (in bits, solid line), and peripheral information variance, $v_{\min}(\cSo^{\gamma}) = v_{\max}(\cSo^{\gamma})$ (in bits$^2$, dashed line), as a function of $\gamma$.\label{fig:ad2}]{\includegraphics[scale=1.32]{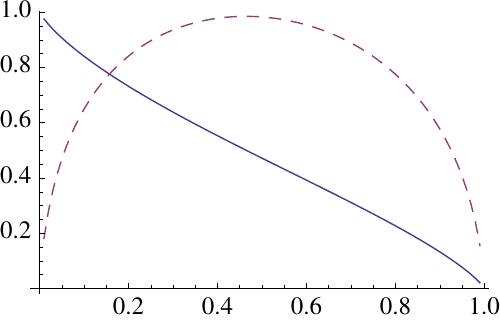}}
  \hspace{1cm}
  \subfigure[\,Second-Order approximation in~\eqref{eq:cor} for $\frac{1}{n} \log M_{\textrm{sep}}^*\big( (\cE_{\textrm{ad}}^{\gamma})^{\otimes n}, \eps \big)$ for $\eps = 1\%$, $\gamma \in \{0, \frac14, \frac34\}$ (top to bottom) as a function of $n$. The dashed lines correspond to the asymptotic limit.\label{fig:ad3}]{\includegraphics[scale=1.32]{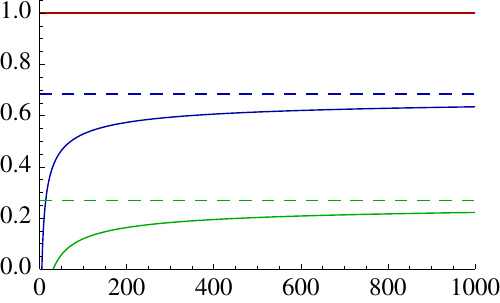}}
  \caption{Geometry and second-order approximation for the amplitude damping channel.}
  \label{fig:ad}
\end{figure}
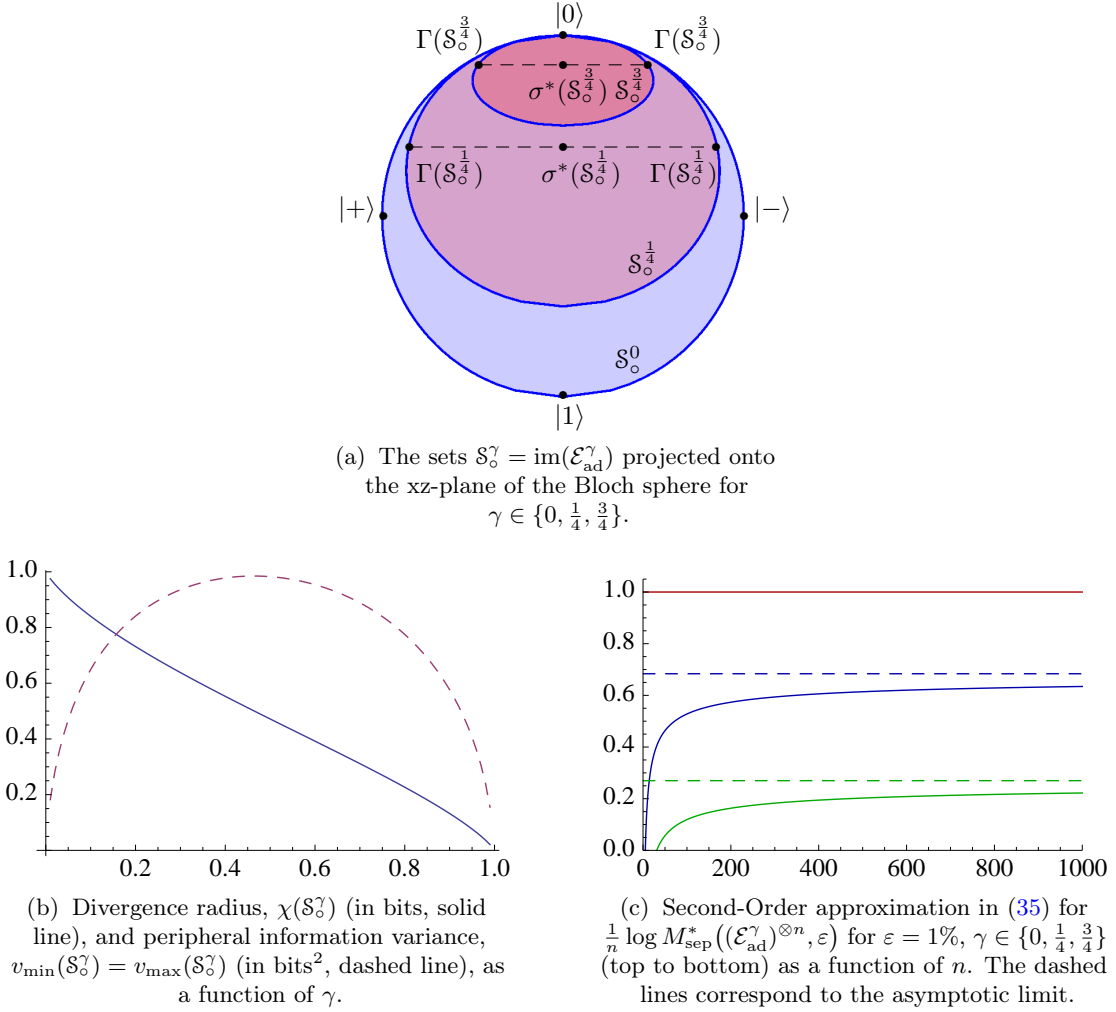

\begin{example}
 The amplitude damping channel with magnitude $\gamma$ is given as
 \begin{align}
 \cE_{\textrm{ad}}^{\gamma}: \rho \mapsto \left( \begin{array}{cc} 1 & 0 \\ 0 & \sqrt{1-\gamma} \end{array} \right) \rho \left(  \begin{array}{cc} 1 & 0 \\ 0 & \sqrt{1-\gamma} \end{array} \right) + \left(  \begin{array}{cc} 0 & \sqrt{\gamma} \\ 0 & 0 \end{array} \right) \rho \left( \begin{array}{cc} 0 & 0 \\ \sqrt{\gamma} & 0 \end{array} \right) .
\end{align}
 Its channel image, $\cSo^{\gamma} = \im(\cE_{\textrm{ad}}^{\gamma})$, is displayed in Figure~\ref{fig:ad1}. In Fig.~\ref{fig:ad2}, the channel capacity and dispersion are evaluated numerically for different values of $\gamma$. The second-order approximation, i.e.\ the first two terms on the right-hand side~of~\eqref{eq:cor} are plotted as a function of $n$ in Figure~\ref{fig:ad3}.
 
 It was already noted in~\cite[Fig.~1]{schumacher01} that it is necessary to consider non-orthogonal input states to achieve $\chi(\cSo^{\gamma})$\,---\,in particular, $\cE_{\textrm{ad}}^{\gamma}(|0\rangle\!\langle0|) \notin \Gamma(\cSo^{\gamma})$ for general $\gamma \in (0,1)$.
\end{example}

This naturally leaves many open questions. Most intriguingly, it was recently shown that for entanglement-breaking and Hadamard channels, we have~\cite{wilde13}
\begin{align}
  \log M^*(\cE^n, \eps) = n\, \chi(\cSo) + O(\sqrt{n})
\end{align}
Thus, one could reasonably conjecture that a second-order approximation of the form~\eqref{eq:cor} also holds for such channels (and not only image-additive channels). In particular, it would be interesting to see if the second-order term is again given by the peripheral information variance. The proof of the strong converse in~\cite{wilde13} relies on the additivity of a suitable R\'enyi divergence radius~\cite{lennert13,wilde13} of the channel image. However, it appears that their techniques are insufficient to derive a second-order expansion of the $\eps$-hypothesis testing divergence radius.

\section{Proofs: Quantum Divergence Radius}\label{sec:radius}

This section contains various lemmas which, combined, establish Theorem~\ref{th:radius}.
Recall that $\cS$ denotes the set of quantum states on a Hilbert space of dimension $d$, and $\cSo \subseteq \cS$ is an arbitrary closed subset of $\cS$, and thus also compact.

We will later show that the divergence center $\sigma^*(\cSo)$, as defined in Theorem~\ref{th:radius}, is indeed a singleton, but at this point we have to be satisfied with the following statement.

\begin{lemma}
  \label{lm:center}
  The set $\sigma^*(\cSo)$ is nonempty, convex and $\sigma \in \sigma^*(\cSo)$ implies $\sigma \gg \rho$ for all $\rho \in \cSo$.
\end{lemma}

\begin{proof}
  Since $\cSo$ is compact, $\sup_{\rho \in \cSo} D ( \rho \| \sigma )$ is finite if and only if $\sigma \gg \rho$ for all $\rho \in \cSo$.
  Moreover, since $\cS$ is compact and the function $f: \sigma \mapsto \sup_{\rho \in \cSo} D ( \rho \| \sigma )$ convex, the set of minima contains at least one element and is convex.
\end{proof}

In analogy to Theorem~\ref{th:radius}, we define the set of extremal points in $\cSo$ corresponding to the center $\sigma \in \sigma^*(\cSo)$ as
$\Gamma_{\sigma}(\cSo) := \argmax_{\rho \in \cSo} D\big(\rho \big\|\sigma \big)$.

\begin{proposition}
  \label{pr:contain}
  For every $\sigma \in \sigma^*(\cSo)$, we have $\sigma \in \conv(\Gamma_{\sigma}(\cSo))$.
\end{proposition}

\begin{proof}
Let us fix $\sigma \in \sigma^*(\cSo)$ to simplify notation. We define 
\begin{align}
\Theta^{\nu} := \big\{ \rho \in \cSo \,\big|\, D(\rho\| \sigma) \geq \chi(\cSo) - \nu \big\} .
\end{align}
and its complement $\bar{\Theta}^{\nu} := \cSo \setminus \Theta^{\nu}$ for any $\nu \geq 0$. We first observe that $\Theta^\nu \subseteq \cSo$ is closed since $D(\cdot\|\sigma)$ is continuous and $\cSo$ is closed itself. Thus, both $\Theta^{\nu}$ and $\conv(\Theta^{\nu})$ are compact.
Moreover, we clearly have $\bigcap_{\nu > 0} \Theta^{\nu} = \Theta^0 = \Gamma_{\sigma}(\cSo)$.

For the sake of contradiction, let us now assume that $\sigma \notin \conv(\Theta^{\nu})$ for some fixed $\nu > 0$. We employ the following lemma (also known as the Pythagorean theorem for relative entropy).

\begin{lemma} \textnormal{\cite[Lem.~3.3]{ohya97}}
\label{lm:closest-point}
  Let $\cSo \subseteq \cS$ be compact convex and let $\sigma \in \cS$. Then, $\tau := \argmin_{\tau \in \cSo} D(\tau\|\sigma)$ is unique. Moreover, for all $\rho \in \cSo$, we have
  \begin{align}
    D(\rho\|\sigma) \geq D(\rho\|\tau) + D(\tau\|\sigma) .
  \end{align}
\end{lemma}

This establishes that there exists a unique state $\tau \in \conv(\Theta^{\nu})$ that minimizes $D(\tau\|\sigma)$. 
Furthermore,  $D(\rho\|\sigma) > D(\rho\|\tau)$ for all $\rho \in \Theta^{\nu}$. Consequently, using the parametrization $\tau^{\lambda} := \lambda \tau + (1-\lambda) \sigma$ and the convexity of $D(\rho\|\cdot)$, we find 
\begin{align}
D(\rho\|\tau^{\lambda}) \leq \lambda D(\rho\|\tau) + (1-\lambda) D(\rho\|\sigma) < D(\rho\|\sigma) \qquad \forall \lambda \in (0, 1) .
\end{align}
Hence, $D(\rho\|\tau^{\lambda}) < D(\rho\|\sigma)$ for all $\rho \in \Theta^{\nu}$ and for all $\lambda \in (0,1)$.

Furthermore, recall that $D(\rho\|\sigma)$ is bounded away from $\chi(\cSo)$ for all $\rho \in \bar{\Theta^{\nu}}$ by definition. 
Due to the continuity of $D(\rho\|\cdot)$, we thus find that for sufficiently small $\lambda > 0$,
\begin{align}
  D\big(\rho \big\| \tau^{\lambda}\big) < \chi(\cSo) \qquad \forall \rho \in \cSo .
\end{align}
However, this implies that $\sigma \notin \sigma^*(\cSo)$ and thus leads to a contradiction.

Hence, we conclude that $\sigma \in \conv(\Theta^{\nu})$ and since this holds for all $\nu > 0$, we find
$\sigma \in \bigcap_{\nu > 0} \conv(\Theta^{\nu})$.
The statement then follows by the following lemma proven in Appendix~\ref{app:set-limit}.
\begin{lemma}
  \label{lm:set-limit}
  Let $\Theta_1 \supseteq \Theta_2 \supseteq \ldots$ be a sequence of compact sets in a finite-dimensional vector space. Then,
  \begin{align}
    \bigcap_{n \in \mathbb{N}} \conv(\Theta_n) = \conv( \Theta_{\infty}) \qquad
    \textrm{whenever}\qquad \Theta_{\infty} := \bigcap_{n \in \mathbb{N}} \Theta_n \neq \emptyset .
  \end{align}
\end{lemma}
\noindent This establishes that
$\bigcap_{\nu > 0} \conv(\Theta^{\nu}) = \conv(\Theta^0)$ and concludes the proof.
\end{proof}

The fact that $\sigma \in \sigma^*(\cSo) \implies \sigma \in \conv(\Gamma_{\sigma}(\cSo))$, first established here, is crucial since it allows the following construction:

Due to Caratheodory's theorem, we may decompose $\sigma$ into a convex combination of (at most $d^2$) peripheral states, namely we may write
\begin{align}
  \label{eq:decomp}
  \sigma = \sum_{\rho \in \cXo} P(\rho)\, \rho, \qquad \textrm{where}\quad  \cXo \subseteq \Gamma_{\sigma}(\cSo),\ |\cXo| \leq d^2 \quad \textrm{and} \quad P \in \cP(\cXo).
\end{align}
Using this decomposition and the fact that $D(\rho\|\sigma) = \chi(\cSo)$ for all $\rho \in \cXo$, we find
\begin{align}
  \label{eq:r-holevo}
  \chi(\cSo) = \sum_{\rho \in \cXo} P(\rho)\, D(\rho \| \sigma) = H(\sigma) - \sum_{\rho \in \cXo} P(\rho)\, H(\rho) \,.
\end{align}
The uniqueness of $\sigma^*(\cSo)$ now follows from a standard argument (see, e.g.,~\cite[Sec.~4.5]{gallager68})
and using the strict concavity of $H$.

\begin{lemma}
  \label{lm:unique}
  The set $\sigma^*(\cSo)$ contains exactly one state.
\end{lemma}

\begin{proof}
  We have already established that $\sigma^*(\cSo)$ is nonempty and convex in Lemma~\ref{lm:center}.
  Assume for the sake of contradiction that $\sigma_0, \sigma_1 \in \sigma^*(\cSo)$ with $\sigma_0 \neq \sigma_1$. Consequently, $\sigma_{\lambda} := \lambda \sigma_1 + (1-\lambda) \sigma_0$ is in $\sigma^*(\cSo)$ for all $\lambda \in [0,1]$. Following~\eqref{eq:decomp}, we may write
  \begin{align}
     \sigma_{\lambda} = \sum_{\rho \in \cXo'} P_\lambda(\rho)\, \rho ,
     \qquad \textrm{where}\quad \cXo' \subseteq \Gamma_{\sigma_0}(\cSo) \cup \Gamma_{\sigma_1}(\cSo),\ |\cXo'| \leq 2d^2\,.
  \end{align}
  and $P_{\lambda}(\rho) = \lambda P_1(\rho) + (1-\lambda) P_0(\rho)$ for $P_{\lambda} \in \cP(\cXo')$
  . Then, due to~\eqref{eq:r-holevo}, we have
  \begin{align}
    \chi(\cSo) &= H(\sigma_0) - \sum_{\rho \in \cXo'} P_0(\rho) H(\rho) = H(\sigma_1) - \sum_{\rho \in \cXo'} P_1(\rho) H(\rho) .
   \end{align}
   Hence, using the strict concavity of $H(\cdot)$, we find
   \begin{align}
    \chi(\cSo) &= \lambda H(\sigma_1) + (1-\lambda) H(\sigma_0) - \sum_{\rho \in \cXo'} P_{\lambda}(\rho)\, H(\rho) \\
    &< H(\sigma_{\lambda}) - \sum_{\rho \in \cXo'} P_{\lambda}(\rho)\, H(\rho) \\
    &= \sum_{\rho \in \cXo'} P_{\lambda}(\rho)\, D(\rho\|\sigma_{\lambda}) .
    \end{align}
    Finally, the fact that $D(\rho\|\sigma_{\lambda}) \leq \sup_{\rho \in \cSo} D(\rho\|\sigma_{\lambda}) = \chi(\cSo)$ since $\sigma_{\lambda} \in \sigma^*(\cSo)$ yields the desired contradiction.
\end{proof}

The previous lemma justifies writing $\Gamma(\cSo)$ in Theorem~\ref{th:radius}, i.e.\ $\Gamma_{\sigma}(\cSo)$ does not depend on $\sigma$. We will thus drop the subscript $\sigma$ in $\Gamma_{\sigma}$ hereafter.

For any $\bbP \in \cP(\cS)$ and $\sigma \in \cS$, let us introduce the notation
\begin{align}
  I(\bbP | \sigma) := \int \diff \bbP(\rho)\, D(\rho \big\| \sigma) \quad \textrm{and} \quad I(\bbP) := I\Big(\bbP \Big| \rho^{(\bbP)} \Big)
\end{align}
in analogy with the conditional mutual information.

\begin{lemma}
  \label{lm:alt}
  We have 
  $\chi(\cSo) = \sup_{\bbP \in \cP(\cSo)} I(\bbP)$. 
  The supremum is achieved by a discrete probability measure with support on at most $d^2$ points in $\Gamma(\cSo)$. 
\end{lemma}

\begin{proof}
  First, note that for every $\bbP \in \cP(\cS)$ we have
     $I(\bbP) = \min_{\sigma \in \cS} I(\bbP|\sigma)$
  due to the positive-definiteness of $D(\cdot\|\cdot)$. Now, Sion's minimax theorem~\cite{sion58} yields
  \begin{align}
     \sup_{\bbP \in \cP(\cSo)} \min_{\sigma \in \cS} I(\bbP|\sigma) = 
     \min_{\sigma \in \cS} \sup_{\bbP \in \cP(\cSo)}  I(\bbP|\sigma) \label{eq:minimax}
  \end{align}
  Indeed, it is easy to verify that $I(\bbP|\sigma)$ is convex in $\sigma$ and linear in $\bbP$. Moreover, $\cS$ is compact convex and $\cP(\cSo)$ is convex, as required. Finally, the supremum over distributions on the right-hand side of~\eqref{eq:minimax} can be replaced by a supremum over Dirac measures on $\cSo$ without loss of generality. This establishes 
  \begin{align}
    \sup_{\bbP \in \cP(\cSo)} \int \diff \bbP(\rho)\, D(\rho \big\| \rho^{(\bbP)}) = \min_{\sigma \in \cS} \sup_{\rho \in \cSo} D(\rho\|\sigma) .
  \end{align}
  
  The second statement follows immediately due to the construction given in Eq.~\eqref{eq:decomp} and~\eqref{eq:r-holevo}.
\end{proof}

We are now ready to summarize the proof of Theorem~\ref{th:radius}.

\begin{proof}[Proof of Theorem~\ref{th:radius}] Property 1 follows from Lemmas~\ref{lm:center} and~\ref{lm:unique}. Property 2 is a trivial consequence of Property 1 and the definition of $\Gamma$. Property 3 is implied by Proposition~\ref{pr:contain} whereas Property 4 is established in Lemma~\ref{lm:alt}. Finally, Property 5 is established as follows:

Clearly, every $\bbP \in \Pi(\cSo)$ achieves the supremum in~\eqref{eq:holevo}, $\chi(\cSo)$, by definition of $\Gamma(\cSo)$. Conversely, let us assume that there exists a distribution $\bbP \in \cP$ that achieves $\chi(\cSo)$. Then, $\rho^{(\bbP)} = \sigma^*(\cSo)$ by the argument in Lemmas~\ref{lm:unique} and~\ref{lm:alt}. Moreover, $\bbP[\Gamma(\cSo)] = 1$ is necessary due to the definition of $\Gamma$.
\end{proof}

\section{Proofs: Second-Order Approximation}\label{sec:asymp}

The direct part of the proof of Theorem~\ref{th:main} is presented in Section~\ref{sec:direct}.
We split the proof of the converse part of Theorem~\ref{th:main} into several parts. First, Section~\ref{sec:proof/one-shot-converse} provides a proof of our one-shot converse bound in Proposition~\ref{pr:one-shot-converse}. Then, Section~\ref{sec:prop} introduces some non-asymptotic bounds on the $\eps$-hypothesis testing divergence for product states that are essential for our asymptotic analysis. As a warm-up, Section~\ref{sec:warm-up} shows the strong converse property for c-q channels using these techniques. The converse part of Theorem~\ref{th:main} is then established in Section~\ref{sec:converse}, and an improved third-order bound for discrete classical-quantum channels is given in~\ref{sec:converethird}.

\subsection{Proof of Direct Part of Thoerem~\ref{th:main}}\label{sec:direct}

We base our result on the following straightforward generalization of the one-shot bounds by Hayashi and Nagaoka~\cite{hayashi03} in the form of Wang and Renner~\cite{wang10} (see also~\cite{datta11a,renes10-2,dupuisszehr12} for recent one-shot achievability bounds for c-q channels).\footnote{To compare with~\cite[Thm.~1]{wang10}, simply note that we may restrict our channel to a discrete classical-quantum channel bijectively mapping from an arbitrary index set to element in $\cXo$. The direct sum notation reveals the classical quantum structure of the underlying state. Finally, the constant $c$ in~\cite{wang10} can be optimized over.}

\begin{proposition}
 \textnormal{\cite[Thm.~1]{wang10}}
  \label{pr:one-shot-direct}
  Let $\eps \in (0, 1)$, $\eta \in (0, \eps)$, and let $\cXo \subseteq \im(\cW)$ be discrete. Then,
  \begin{align}
    \log M^*(\cW, \eps) \geq 
  &\sup_{P \in \cP(\cXo)} D_h^{\eps-\eta}\Bigg( \bigoplus_{\rho \in \cXo} P(\rho) \rho\, \Bigg\| \bigoplus_{\rho \in \cXo} P(\rho) \rho^{(P)} \Bigg) - \log \frac{4\eps (1-\eps+\eta)}{\eta^2}  .
  \end{align}
\end{proposition}

 We can include the closure of $\im(\cW)$ due to the continuity of the above expression when the set $\cXo$ is varied by replacing an element with one that is close in $(\cS,\delta_{textrm{tr}})$. Thus, our bound reads
\begin{align}
  \log M^*(\cW, \eps) \geq \sup_{\cXo \subseteq \overline{\im(\cW)}} \sup_{P \in \cP(\cXo)} D_h^{\eps-\eta}\Big( \omega^{(P)} \Big\| \tau^{(P)} \otimes \rho^{(P)}  \Big) - \log \frac{4\eps (1-\eps+\eta)}{\eta^2} , \label{eq:direct}
\end{align}
where $\cXo$ is discrete and we introduced the shorthands $\omega^{(P)} := \bigoplus_{\rho \in \cXo} P(\rho) \rho$ and $\tau^{(P)} := \bigoplus_{\rho \in \cXo} P(\rho)$.
The similarity of the above expression with the asymptotic expression in~\eqref{eq:holevo} is evident once~\eqref{eq:holevo} is specialized to the discrete case as well.

The restriction to finite subsets of $\cSo$ is unproblematic in light of Lemma~\ref{lm:cara}.
Let us then proceed to prove the lower bound in Theorem~\ref{th:main}, which we restate in a slightly stronger form here.

\begin{directpart}
  Let $\eps \in (0, 1)$ and let $\cW$ be a c-q channel. Set $\cSo := \overline{\im(\cW)}$. Then,
  \begin{align}
    \log M^*(\cW^n, \eps) \geq n\, \chi(\cSo) + \sqrt{n\, v_{\eps}(\cSo)}\,\Phi^{-1}(\eps) + O(\log n) .
  \end{align}
\end{directpart}

\begin{proof}
  First, let us apply~\eqref{eq:direct} to the $n$-fold repetition of the channel $\cW$. Fixing any discrete set $\cXo \subseteq \cSo$ and $P \in \cP(\cXo)$, we first confirm that 
  \begin{align}
    \log M^*(\cW^n, \eps) \geq D_h^{\eps-\eta} \Big( \big( \omega^{(P)} \big)^{\otimes n}  \Big\| \big( \tau^{(P)} \otimes \rho^{(P)} \big)^{\otimes n} \Big) - \log \frac{4\eps (1-\eps+\eta)}{\eta^2} \label{eq:direct1}
  \end{align}
  Note that we applied~\eqref{eq:direct} using the set $\cXo^{\otimes n} \subseteq \overline{\im(\cW)}^{\otimes n} = \overline{\im(\cW)^{\otimes n}}$ and the $n$-fold product distribution $P^{\times n}$. By Lemma~\ref{lm:cara} there exists a probability mass function (let it be our choice of $P$) with support on $\Gamma(\cSo)$ (let the support set be our choice of $\cXo$) such that
  \begin{align}
    \rho^{(P)} = \sigma ,  \quad \sum_{\rho \in \cXo} P(\rho)\, D(\rho\|\sigma) = \chi(\cSo), \quad \textrm{and} \quad \sum_{\rho \in \cXo} P(\rho)\, V(\rho\|\sigma) = v_{\eps}(\cSo) , \label{eq:direct3}
  \end{align}
  where we set $\sigma = \sigma^*(\cSo)$. Now, we can verify that
  \begin{align}
    \sum_{\rho \in \cXo} P(\rho)\, D\big(\rho\big\|\rho^{(P)}\big) = D\Bigg( \bigoplus_{\rho \in \cXo} P(\rho) \rho\, \Bigg\| \bigoplus_{\rho \in \cXo} P(\rho) \rho^{(P)} \Bigg) , \label{eq:direct2}
  \end{align}
  and, the following simple generalization of~\cite[Lm.~62]{polyanskiy10} proved in Appendix~\ref{app:direct} holds.
  \begin{lemma}
     \label{lm:v-expand} For any probability mass function $P \in \Pi(\cSo)$, we have
     \begin{align}
         \sum_{\rho \in \cXo} P(\rho)\, V\big(\rho\big\|\rho^{(P)}\big) = V\Bigg( \bigoplus_{\rho \in \cXo} P(\rho) \rho\, \Bigg\| \bigoplus_{\rho \in \cXo} P(\rho) \rho^{(P)} \Bigg) .
     \end{align}
  \end{lemma}
  As such, we are left to evaluate the asymptotics of $D_h^{\eps -\eta}$ for identical product states. 
  Let us set $\eps_n := \eps - \eta$ with $\eta = 1/{\sqrt{n}}$.
  First, consider the case where $v_{\eps}(\cSo) > 0$. The second assertion in Proposition~\ref{pr:asymptotics} (i.e.\ the bound in~\eqref{eq:berry}) specialized to i.i.d.\ states, establishes that
    \begin{align}
      &D_h^{\eps-\eta} \Big( \big( \omega^{(P)} \big)^{\otimes n}  \Big\| \big( \tau^{(P)} \otimes \rho^{(P)} \big)^{\otimes n} \Big) \\
      &\qquad \qquad  \geq n D\big( \omega^{(P)} \big\| \tau^{(P)} \otimes \rho^{(P)} \big) + \sqrt{n V\big( \omega^{(P)} \big\| \tau^{(P)} \otimes \rho^{(P)} \big) }\, \Phi^{-1}(\eps) - L_2 \log n \\
      &\qquad \qquad  = n \chi(\cSo) + \sqrt{n\, v_{\eps}(\cSo)}\,\Phi^{-1}(\eps) - L_2 \log n \,.
    \end{align}
    for all $n \geq N_2$ and some constants $L_2$ and $N_2(\eps,\cSo)$.
    In the last step we employed~\eqref{eq:direct2}, Lemma~\ref{lm:v-expand} and~\eqref{eq:direct3}. Moreover, the last summand in~\eqref{eq:direct1} is of the form $O(\log n)$ and we are done.
  
  The proof for the case $v_{\eps}(\cSo) = 0$ proceeds similarly but employs the first bound in Eq.~\eqref{eq:cheby} in Proposition~\ref{pr:asymptotics} instead. This yields
  \begin{align}
    D_h^{\eps-\eta} \Big( \big( \omega^{(P)} \big)^{\otimes n}  \Big\| \big( \tau^{(P)} \otimes \rho^{(P)} \big)^{\otimes n} \Big) \geq n D\big( \omega^{(P)} \big\| \tau^{(P)} \otimes \rho^{(P)} \big) - L_1 \log n .
  \end{align}
  for all $n \geq N_1$ and some constants $L_1$ and $N_1(\eps-\eta,\cSo)$. The rest of the proof then proceeds analogously to the discussion above.
\end{proof}

\subsection{Proof of Proposition~\ref{pr:one-shot-converse}}
\label{sec:proof/one-shot-converse}

Let us recall the statement of Proposition~\ref{pr:one-shot-converse}. For $\eps \in (0,1)$ and $\mu \in (0, 1-\eps)$, we want that
  \begin{align}
  \log M^*(\cW,\eps ) \leq \chi_h^{\eps+\mu}\Big(\overline{\im(\cW)}\Big) + \log \frac{\eps+\mu}{\mu(1-\eps-\mu)} . \label{eq:conv2}
  \end{align}

\begin{proof}[Proof of Proposition~\ref{pr:one-shot-converse}]
  Let $\cC =\{ \cM, e, \cD \}$ be a code with $p_{\textrm{err}}(\cC, \cW) \leq \eps$ given by codewords $x_m = e(m) \in \cX$ and a decoder $\cD = \{Q_m\}_{m\in\cM}$. By assumption, we thus have $\frac{1}{|\cM|} \sum_{m \in \cM} \tr(Q_m \cW(x_m) ) \geq 1 - \eps$. For an arbitrary but fixed $\sigma \in \cS$, we define the set 
\begin{align}
  \cK := \big\{ m \in \cM \,\big| \tr\big( Q_m \cW(x_m) \big) \geq 1 - \eps - \mu \big\}, \quad \textrm{and} \quad m^* := \argmin_{m \in \cK}\ \tr( Q_m \sigma ).
\end{align} 
By definition of this set, we have
\begin{align}
  1 - \eps \leq \frac{1}{|\cM|} \sum_{m \in \cM} \tr( Q_m \cW(x_m) ) &= \frac{1}{|\cM|} \sum_{m \in \cK} \tr( Q_m \cW(x_m) ) +
   \frac{1}{|\cM|} \sum_{m \in \cM \setminus \cK} \tr( Q_m \cW(x_m) ) \\
    &< \frac{|\cK|}{|\cM|} + \frac{|\cM| - |\cK|}{|\cM|} (1 - \eps - \mu) 
\end{align}
  Hence, $|\cK| > |\cM| \frac{\mu}{\eps+\mu}$. Moreover, we have
  \begin{align}
    1 = \tr(\sigma) = \sum_{m \in \cM} \tr(Q_{m} \sigma) \geq |\cK| \tr(Q_{{m^*}} \sigma) > |\cM| \frac{\mu}{\eps+\mu} \tr(Q_{{m^*}} \sigma) .
  \end{align}
  By definition of the $\eps$-hypothesis testing divergence we find
  \begin{align}
     D_h^{\eps+\mu}(\cW(x_{m^*}) \| \sigma) \geq - \log \frac{ \tr(Q_{{m^*}} \sigma)}{1-\eps-\mu} > \log |\cM| - \log \frac{ \eps + \mu }{\mu(1-\eps-\mu)} .
  \end{align}
  Thus, in particular we have
  \begin{align}
    \sup_{\rho \in\, \overline{\im(\cW)}}\, D_h^{\eps+\mu}(\rho \| \sigma) > \log |\cM| - \log \frac{ \eps + \mu }{\mu(1-\eps-\mu)}
  \end{align}
  Finally, Eq.~\eqref{eq:conv2} follows by observing that the above bound holds for all $\sigma \in \cS$.
\end{proof}

\subsection{Non-Asymptotic Bounds on the Hypothesis-Testing Divergence}\label{sec:prop}

Some of the main ingredients of our asymptotic analysis in the converse part of the proof of Theorem~\ref{th:main} are the following non-asymptotic bounds on the $\eps$-hypothesis testing divergence evaluated for product states.  
Before we state the bounds, recall that
$I(\bbP | \sigma) = \int \diff \bbP(\rho) \, D(\rho\|\sigma)$ and define $V(\bbP | \sigma) := \int \diff \bbP(\rho) \, V(\rho\|\sigma)$ analogously for any $\bbP \in \cP(\cS)$ and $\sigma \in \cS$. Moreover, given a sequence of states $\rho^n = \bigotimes_{i=1}^n \rho_i$, we denote by $P_{\rho^n}(\rho) := \frac1{n} \sum_{i=1}^n 1\{\rho = \rho_i\}$ the empirical distribution of $\rho^n$.

\begin{proposition}
  \label{pr:asymptotics}
  Let $\eps \in (0,1)$, $\cSo \subseteq \cS$ and $\lambda_0 > 0$. Let $\{\eps_n\}_{n=1}^{\infty}$ be any sequence satisfying $|\eps_n - \eps| \leq 1/\sqrt{n}$ for all $n$ and set $\eps^* := \min\{\eps, 1-\eps\}$. Then, there exist constants $N_1(\eps,\cSo,\lambda_{0})$ and $K_1(\eps,\cSo,\lambda_{0})$ and $L_1$ such that the following holds. For every $n \geq N_1$, every $\sigma \in \cS$ with $\lambda_{\min}(\sigma) \geq \lambda_0$ and every sequence $\rho^n = \bigotimes_{i=1}^n \rho_i$,  $\rho_i \in \cSo$, we have
  \begin{align}
     \Big|\, D_h^{\eps_n}\big(\rho^n \big\|\, \sigma^{\otimes n} \big) - n I \big( P_{\rho^n} \big| \sigma\big) \Big| 
     \leq \sqrt{\frac{n V(P_{\rho^n}|\sigma)}{\eps^*}} + L_1 \log n
     \leq K_1 \sqrt{n}\, . \label{eq:cheby}
  \end{align}

  Further let $\xi > 0$ and fix $\sigma \in \cS$ with $\lambda_{\min}(\sigma) > 0$. Then, there exist constants $N_2(\eps,\cSo,\sigma,\xi)$ and $L_2$ such that the following holds. For every $n \geq N_2$ and every sequence $\rho^n = \bigotimes_{i=1}^n \rho_i$,  $\rho_i \in \cSo$ satisfying $V(P_{\rho^n}|\sigma) \geq \xi$, we have
  \begin{align}
  \Big|\, D_h^{\eps_n}\big(\rho^n \big\|\, \sigma^{\otimes n} \big) - n I \big( P_{\rho^n} \big| \sigma\big) - \sqrt{n V\big( P_{\rho^n} \big| \sigma \big)}\,\Phi^{-1}(\eps)  \Big| \leq L_2 \log n . \label{eq:berry}
  \end{align}
  Finally, if $\sigma = \rho^{(P_{\rho^n})}$ in~\eqref{eq:berry}, then the statement holds for
  $n \geq N_3(\eps,\cSo,\xi)$ independent of $\sigma$.
\end{proposition}

In the asymptotic limit as $n \to \infty$, all inequalities imply the seminal quantum Stein's lemma~\cite{hiai91} and its strong converse~\cite{ogawa00} when the sequence is chosen i.i.d. The proof is based on the techniques of~\cite{li12,tomamichel12} and presented in Appendix~\ref{app:hypo}.
It is crucial for our application that $L_1, L_2, K_1, N_1, N_2$ and $N_3$ are uniform over $\sigma$ and sequences $\rho^n$ satisfying the constraints. This is nontrivial and requires arguments beyond those in~\cite{li12,tomamichel12} which only treat the i.i.d.\ case.\footnote{For this reason we also do not rely on the 
ubiquitous $O(\cdot)$ notation here, which tends to hide such subtleties.}

\subsection{Asymptotics of the $\eps$-Hypothesis Testing Divergence Radius: First-Order}
\label{sec:warm-up}

As a warm-up, we use our techniques to provide a simple proof of the strong converse property of general classical-quantum channels. The strong converse is evidently a corollary of Proposition~\ref{pr:one-shot-converse} and the following result.\footnote{To verify this, apply Proposition~\ref{pr:one-shot-converse} for the $n$-fold repetition of the channel, $\cW^n$ with image $\cSo^{\otimes n}$, and choose $\mu(n) = 1/\sqrt{n}$ such that $\eps_n = \eps + \mu$ in Proposition~\ref{pr:strong}.}

\begin{proposition}
  \label{pr:strong}
  Let $\eps \in (0,1)$ and $\cSo \subseteq \cS$ closed. Let $\{\eps_n\}_{n=1}^{\infty}$ be any sequence satisfying $|\eps_n - \eps| \leq 1/\sqrt{n}$ for all $n$. Then,
  \begin{align}
    \chi_h^{\eps_n}\big( \cSo^{\otimes n} \big) \leq n\, \chi(\cSo) + O(\sqrt{n}) .
  \end{align}
\end{proposition}

Note that Winter~\cite{winterthesis} and Ogawa-Nagaoka~\cite{ogawa99} first showed the strong converse for classical-quantum channels for the generality we consider here.

\begin{proof}
By definition of the $\eps$-hypothesis testing divergence radius, we have
\begin{align}
 \chi_h^{\eps_n}\big(\cSo^{\otimes n}\big) \leq \sup_{\rho^n \in \cSo^{\otimes n}} D_h^{\eps_n} \big( \rho^n \,\big\|\, \sigma^{\otimes n} \big) , \label{eq:chiheps-bound}
\end{align}
where we chose an $n$-fold product of the divergence center, $\sigma = \sigma^*(\cSo) \in \cS$, as the output state.
The states $\rho^n$ are of the form $\rho^n = \bigotimes_{i=1}^n \rho_i$. For a fixed and arbitrary $\rho^n$, we define the set $\cSo^n := \{ \rho_i \}_{i=1}^n \subseteq \cSo$ and the empirical distribution $P_{\rho^n} \in \cP\big(\cSo^n\big)$ given by $P_{\rho^n}(\rho) = \frac1n \sum_{i=1}^n 1\{ \rho = \rho_i \}$.

We then use~\eqref{eq:cheby} in Proposition~\ref{pr:asymptotics} to assert that 
\begin{align}
D_h^{\eps_n} \big( \rho^n \,\big\|\, \sigma^{\otimes n} \big) \le n I(P_{\rho^n} | \sigma) + K_1 \sqrt{n} \label{eq:sc-1}
\end{align}
for sufficiently large $n \geq N_1$. Here, we used that $\lambda_{\min}(\sigma) > 0$ and recall that $I(\bbP|\sigma)$ is defined as
$I(\bbP|\sigma) = \int \diff \bbP(\rho)\, D( \rho \| \sigma)$.
Therefore, Theorem~\ref{th:radius} ensures that $D(\rho\|\sigma) \leq \chi(\cSo)$ for all $\rho \in \cS_o$ and we have established that
\begin{align}
  &\chi_h^{\eps_n}\big(\cSo^{\otimes n}\big) \leq \sup_{\rho^n \in \cSo^{\otimes n}} n I(P_{\rho^n}|\sigma) + K_1 \sqrt{n} \leq n \chi(\cSo) + K_1 \sqrt{n} . \qedhere
\end{align}
\end{proof}

\subsection{Asymptotics of the $\eps$-Hypothesis Testing Divergence Radius: Second-Order}\label{sec:converse}

In view of Proposition~\ref{pr:one-shot-converse} and the discussion in the previous section, we therefore want to find a second-order upper bound on 
$\chi_h^{\eps}\big(\cSo^{\otimes n}\big) = \min_{\sigma^n \in\cS^{n}}
 \sup_{\rho^n \in \cSo^{\otimes n}} D_h^{\eps}\big( \rho^n \| \sigma^n\big) $.
The following results constitute the main technical contribution of this paper.

\subsubsection{An Appropriate Choice of $\sigma^n$}

The proof of the strong converse in Propositon~\ref{pr:strong} hinges on choosing $\sigma^n$ as the $n$-fold product of the divergence center and then taking advantage of the fact that $D(\rho\|\sigma) \leq \chi(\cSo)$ for all $\rho \in \cSo$. This will not be sufficient if we want to pin down the exact second-order term proportional to $\sqrt{n}$.\footnote{To see why this is so, consider a sequence of states $\rho^n = \bigotimes_{i=1}^n \rho_i$ with $\rho_i \in \Gamma(\cSo)$. Then, following the notation in the proof of Proposition~\ref{pr:strong}, we realize that $D_n = \chi(\cSo)$. However, since $\frac1n \sum_{i=1}^n \rho_i \neq \sigma^*(\cSo)$ in general, the empirical distribution $P_{\rho^n}$ can be arbitrarily far from $\Pi(\cSo)$. Thus, we cannot hope to bound $V_n$ in terms of $v_{\eps}$.}

Before we commence, we thus introduce an appropriate choice of auxiliary state $\sigma^n$. To construct it, we require the following auxiliary result whose proof is provided in Appendix~\ref{app:net}. This establishes that there exists a $\gamma$-net on $\cSo$ whose cardinality can be bounded appropriately.

\begin{lemma} \label{lm:net}
  For every $\gamma \in (0,1)$, there exists a set of states $\cG^{\gamma} \subseteq \cS$ of size 
  \begin{align}
  |\cG^\gamma| \leq \left(\frac{5}{\gamma}\right)^{2d^2} \left(\frac{2d}{\gamma}+2\right)^{d-1}
  \end{align}
   such that, for every $\rho \in \cS$, there exists a state $\tau \in \cG^\gamma$ satisfying the following:
  \begin{align}
    \frac{1}{2} \| \rho - \tau \|_1 \leq \gamma, \quad 
    D(\rho\|\tau) \leq \gamma \cdot 4 (2d + 1), \quad \textrm{and} \quad
    \lambda_{\min}(\tau) \geq \frac{\gamma}{2d+\gamma} .
  \end{align}
\end{lemma} 
%
Now, for a $\gamma$ to be specified below, we choose the output state $\sigma^n \in\cS^{n}$ as follows:
\begin{align}
\sigma^n := \frac12 \sigma^{\otimes n} + \frac1{2 |\cG^\gamma|} \sum_{\tau \in   \cG^\gamma } \tau^{\otimes n} , \qquad \textrm{where}\quad \sigma = \sigma^*(\cSo). \label{eq:sigman}
\end{align}
Note that $\sigma^n$ is normalized  and   is, in fact,   a convex combination of the $n$-fold tensor product of the divergence center and the $n$-fold tensor product of the  elements of the net, of which there are only finitely many. With this choice of $\sigma^n$ we bound 
$D_h^{\eps}\big( \rho^n \| \sigma^n \big)$
in the following. 

\subsubsection{Different Sequences of Inputs}

We will also need to treat different types of state sequences separately. We keep $\cSo$ fixed for the following to simplify notation. Let us define $\Omega_1^{\nu}, \Omega_2^{\nu} \subseteq \cSo^{\otimes n}$ for some $0 < \nu \leq 1$, which describe sets of state sequences of length $n$ that are close to achieving the first-order fundamental limit. (We omit the dependence on $n$ in our notation here.) The first set ensures that the states are close to $\Gamma(\cSo)$, and is defined as
\begin{align}
 \Omega_1^{\nu} := \Bigg\{ \rho^n \in \cSo^{\otimes n} \ \Bigg|\ \frac{1}{n} \sum_{i=1}^n \underbrace{\min_{\tau \in \Gamma(\cSo)} \frac12 \| \rho_i - \tau \|_1}_{ =:\, \Delta(\rho_i, \Gamma(\cSo))} \leq \nu \Bigg\} .
\end{align}
The second set ensures that the average state is close to the divergence center, and is defined as
\begin{align}
  \Omega_2^{\nu} := \Bigg\{ \rho^n \in \cSo^{\otimes n} \ \Bigg|\ \frac12 \bigg\| \frac1{n} \sum_{i=1}^n \rho_i - \sigma^*(\cSo) \bigg\|_1 \leq \nu \Bigg\} .
\end{align}
The interesting, close to capacity-achieving sequences are those that are in $\Omega_1^{\nu} \cap \Omega_2^{\nu}$.

\subsubsection{Dealing with Sub-Optimal Input Sequences}

We first deal with sequences that are far from optimal in the sense prescribed above.
\begin{proposition}
  \label{pr:conv-bad}
  Let $\eps \in (0,1)$, $\nu > 0$ and $\cSo \subseteq \cS$. Let $\{\eps_n\}_{n=1}^{\infty}$ be any sequence satisfying $|\eps_n - \eps| \leq 1/\sqrt{n}$ for all $n$. Then, there exist constants $N_0(\eps, \cSo, \nu)$ and $\gamma_0(\cSo,\nu)$ such that, for all $n \geq N_0$ and all $\rho^n \notin \Omega_1^{\nu} \cap \Omega_2^{\nu}$, we have
  \begin{align}
     D_h^{\eps_n}(\rho^n \| \sigma^n) 
       \leq n\, \chi(\cSo) + \sqrt{n\, v_{\eps}(\cSo)}\,\Phi^{-1}(\eps) ,
   \end{align}
   where $\sigma^n$ is defined as in~\eqref{eq:sigman} for a fixed $\gamma = \gamma_0$.
\end{proposition}

\begin{proof}
The technique for bounding $D_h^{\eps_n}\big( \rho^n \| \sigma^n \big)$ differs depending  on the  state sequence $\rho^n$. We consider two  cases: (a) $\rho^n \notin \Omega_1^{\nu}$ and (b) $\rho^n \notin \Omega_2^{\nu}$ in the following subsections. 

\vspace{0.2cm} 
\paragraph*{(a) Sequences $\rho^n \notin \Omega_1^{\nu}$:}

Applying Property 3 of Lemma~\ref{lm:hypo-prop} to $D_h^{\eps_n}(\rho^n\|\sigma^n)$ with our choice of $\sigma^n$  in \eqref{eq:sigman} and picking out the divergence center $\sigma^{\otimes n}$ yields an upper bound of the form
\begin{align}
D_h^{\eps_n}(\rho^n\|\sigma^n) &\leq  D_h^{\eps_n} (\rho^n \| \sigma^{\otimes n} \big) + \log 2. 
\end{align}
Furthermore, as in the proof of Proposition~\ref{pr:strong}, we employ~\eqref{eq:cheby} in Proposition~\ref{pr:asymptotics} to obtain
\begin{align}
D_h^{\eps_n}\big(\rho^n\big\|\sigma^{\otimes n}\big) &\leq \sum_{i=1}^n D( \rho_i \| \sigma) + K_1 \sqrt{n} \, ,
\end{align}
for all $n \geq N_1$. (We absorbed the constant term $\log 2$ into the constant $K_1$ here for convinience.)

Now, we define
$\hat{\chi}^{\nu}_1 := \sup_{ \rho \in \cSo :\, \Delta(\rho,\Gamma) > \frac{\nu}2 } D ( \rho \| \sigma) < \chi(\cSo)$ and
employ the following lemma which is shown in Appendix~\ref{app:selection}.
\begin{lemma} \label{lm:selection} 
  Let $\rho^n \in \cSo^{\otimes n}$ be fixed and let $\nu \in (0, 1)$. If $\rho^n \notin \Omega_1^{\nu}$, then   there exists a set $\Xi^{\nu} \subseteq [n]$ of cardinality $| \Xi^{\nu} | > n \frac{\nu}{2}$ such that, for all $i \in \Xi^{\nu}$, we have $\Delta(\rho_i, \Gamma) > \frac{\nu}{2}$.
\end{lemma}
\noindent This leads us to bound
\begin{align}
D_h^{\eps_n}(\rho^n\|\sigma^n) \leq \sum_{i \in \Xi^{\nu}} \hat{\chi}^{\nu}_1 + \sum_{i \notin \Xi^{\nu}} \chi(\cSo) +  K_1 \sqrt{n} 
\leq n \chi(\cSo) - n(\chi(\cSo) -\hat{\chi}^{\nu}_1) \frac{\nu}{2} + K_1 \sqrt{n} . 
\end{align}
In particular, we have $D_h^{\eps_n}(\rho^n\|\sigma^n) \le n\,\chi(\cSo) + \sqrt{n\, v_\eps(\cSo)} \Phi^{-1}(\eps)$ for sufficiently large $n \geq N$, where $N$ is appropriately chosen.

\vspace{0.2cm} 
\paragraph*{(b) Sequences $\rho^n \notin \Omega_2^{\nu}$:}

For these sequences, we extract the state $\tau^{\otimes n}$ from the convex combination that defines $\sigma^n$   in~\eqref{eq:sigman}, where $\tau$ is the state closest (in the relative entropy sense) to the average output state $\bar{\rho} = \rho^{(P_{\rho^n})} = \frac1{n} \sum_{i=1}^n \rho_i$ in $\cG^{\gamma}$  and the constant $\gamma >0$ is to be chosen later. In other words,  $\tau \in\argmin_{\tau\in\cG^\gamma}D(\bar{\rho}\|\tau)$. Thus, by Property 3 of Lemma~\ref{lm:hypo-prop}, we have
\begin{align}
  D_h^{\eps_n}(\rho^n\|\sigma^n) \leq D_h^{\eps_n}\big(\rho^n\big\|\tau^{\otimes n}\big) + \log |\cG^{\gamma}| .
\end{align}

Then, by using~\eqref{eq:cheby} in Proposition~\ref{pr:asymptotics} we find for all $\rho^n \notin \Omega_2^{\nu}$ that 
\begin{align}
D_h^{\eps_n}\big(\rho^n\big\|\tau^{\otimes n}\big) &\leq \sum_{i=1}^n D(\rho_i \| \tau) + K_1' \sqrt{n} .  
\end{align} 
for $n \geq N_1'$. Here, we take advantage of the fact that the minimum eigenvalue of $\tau$ satisfies $\lambda_{\min}(\tau) \geq \frac{\gamma}{2d+\gamma}$ such that the constants $K_1'$ and $N_1'$ can be chosen uniformly for all $\tau \in \cG^{\gamma}$.

We continue to bound
\begin{align}
D_h^{\eps_n}\big(\rho^n\big\|\tau^{\otimes n}\big)
 &\leq \sum_{i=1}^n D( \rho_i \| \bar{\rho}) + \sum_{i=1}^n \tr \big( \rho_i (\log \bar{\rho} - \log \tau) \big) + K_1' \sqrt{n}
  \\
   &= \sum_{i=1}^n D( \rho_i \| \bar{\rho}) + n D(\bar{\rho}\|\tau) +  K_1' \sqrt{n} \\
   &\leq n\, I\Big(P_{\rho^n} \Big|\, \rho^{(P_{\rho^n})} \Big) +n\cdot 4\gamma (2d + 1)+  K_1' \sqrt{n} ,
\end{align}
where the second inequality follows from the properties  of the $\gamma$-net stated in Lemma~\ref{lm:net} 
and on the last line we introduced the empirical distribution of $\rho^n$, defined as $P_{\rho^n}(\rho) = \frac1{n} \sum_{i=1}^n 1 \{ \rho = \rho_i \}$.

 Then, by Theorem~\ref{th:radius} and the definition of $\Pi(\cSo)$ and $\nu \in (0,1)$, we know that 
 \begin{align}
 \tilde{\chi}^{\nu}_2 := \sup  \bigg\{ I\Big( \bbP \Big| \rho^{(\bbP)} \Big) \,\bigg|\, \bbP \in \cP(\cSo):\, \frac12 \Big\| \rho^{(\bbP)} - \sigma^*(\cSo) \Big\|_1 > \nu \bigg\} < \chi(\cSo) . 
 \end{align}
Summarizing the above, we have
\begin{align}
D_h^{\eps_n}(\rho^n\|\sigma^n) \leq n \chi(\cSo) - n \big( \chi(\cSo) - \tilde{\chi}^{\nu}_2 - 4\gamma (2d + 1) \big)+ K_1' \sqrt{n} + \log |\cG^{\gamma}| .
\end{align}
  By choosing $\gamma = \gamma_0(\nu,\cSo)$ small enough such that $\chi(\cSo) -\tilde{\chi}^{\nu}_2
   -  4\gamma (2d + 1) > 0$, we find that $D_h^{\eps_n}(\rho^n\|\sigma^n) \leq n\chi(\cSo) + \sqrt{n v_{\eps}(\cSo)} \Phi^{-1}(\eps)$ for sufficiently large $n \geq N'$, appropriately chosen.
   
  We conclude by observing that the statement of the proposition holds for $n \geq \max\{N, N'\}$. 
\end{proof}

\subsubsection{Putting Everything Together: Proof of Converse Part of Theorem~\ref{th:main}}

The upper bound in Theorem~\ref{th:main} is now a corollary of Proposition~\ref{pr:one-shot-converse} and the following result.


\begin{proposition}
  \label{pr:thedifficultpart}
  Let $\eps \in (0,1)$ and $\cSo \subseteq \cS$. Let $\{\eps_n\}_{n=1}^{\infty}$ be any sequence satisfying $|\eps_n - \eps| \leq 1/\sqrt{n}$ for all $n$. Then,
  \begin{align}
    \chi_h^{\eps_n}(\cSo^{\otimes n}) \leq n\, \chi(\cSo) + \sqrt{n\, v_{\eps}(\cSo)}\, \Phi^{-1}(\eps) + o(\sqrt{n}) .
  \end{align}
\end{proposition}

\begin{proof}
  For any $\nu \in (0, 1)$, we first invoke Proposition~\ref{pr:conv-bad} to verify that
  \begin{align}
    \sup_{\rho^n \notin \Omega_1^{\nu} \cap \Omega_2^{\nu}} D_h^{\eps_n}(\rho^n \| \sigma^n) 
       \leq n\, \chi(\cSo) + \sqrt{n\,v_{\eps}(\cSo)}\,\Phi^{-1}(\eps) \label{eq:outliers}
  \end{align}
  for $n \geq N_0(\eps,\nu,\cSo)$ sufficiently large. It remains to consider sequences $\rho^n \in \Omega_1^{\nu} \cap \Omega_2^{\nu}$.
Define the set  of sequences $\rho^n$ with empirical distribution $P_{\rho^n}$ resulting in a $\xi$-positive relative entropy variance as 
\begin{equation}
\Omega_3^\xi:=\left\{ \rho^n  \in  \cSo^{\otimes n } : V( P_{\rho^n} |\sigma)  \ge \xi\right\}, \label{eq:xiset}
\end{equation}
where $\xi>0$ is a constant to be chosen later. 
%

  For $\rho^n \notin \Omega_3^\xi$, 
  we again pick out $\sigma^{\otimes n}$ from~\eqref{eq:sigman} to find $D_h^{\eps_n}(\rho^n\|\sigma^n) \leq D_h^{\eps_n}(\rho^n\|\sigma^{\otimes n}) + \log 2$.
  Then, we employ~\eqref{eq:cheby} in Proposition~\ref{pr:asymptotics} to obtain
\begin{align}
D_h^{\eps_n}(\rho^n \| \sigma^{\otimes n}) \le n I(P_{\rho^n}|\sigma) + \sqrt{\frac{n V(P_{\rho^n}|\sigma)}{\eps^*}} + L_1 \log n < n \chi(\cSo) + \sqrt{\frac{n\xi}{\eps^*}} + L_1 \log n .  \label{eq:cheby1}
\end{align}
For sequences $\rho^n \in \Omega_3^\xi$, by the Berry-Esseen-type bound~\eqref{eq:berry} in Proposition~\ref{pr:asymptotics}, we have   
\begin{align}
D_h^{\eps_n}(\rho^n \| \sigma^{\otimes n}) &\le n I(P_{\rho^n}|\sigma) + \sqrt{n V(P_{\rho^n}|\sigma)} \Phi^{-1}(\eps)+ L_2 \log n  \nonumber\\
  &\leq n \chi(\cSo) + \sqrt{n\, v_{\eps}^{\nu}(\cSo)}\,\Phi^{-1}(\eps) + L_2 \log n \label{eq:berry1} ,
\end{align}
where we define $v_{\eps}^{\nu}(\cSo)$ similarly to $v_{\eps}(\cSo) = v_{\eps}^0(\cSo)$ as
\begin{align}
v_{\eps}^{\nu}(\cSo) := \begin{cases} \inf_{\bbP \in \Pi^{\nu}} V(\bbP|\sigma) & \textrm{if }\ 0 < \eps \leq \frac12 \\
       \sup_{\bbP \in \Pi^{\nu}} V(\bbP|\sigma) & \textrm{if }\ \frac12 < \eps < 1 \end{cases} 
\end{align}
where we employed the set $\Pi^{\nu} \subseteq \cP(\cSo)$ of probability measures close to $\Pi(\cSo)$, given as
\begin{align}
  \Pi^{\nu} := \bigg\{ \bbP \in \cP(\cSo) \,\bigg|\, \int \diff \bbP(\rho)\, \Delta(\rho, \Gamma) \leq \nu \ \land \ \frac{1}{2}\Big\| \rho^{(\bbP)} - \sigma^*(\cSo) \Big\|_1 \leq \nu \bigg\} .
\end{align}
Clearly, the empirical distribution of a sequence $\rho^n$ is in $\Pi^{\nu}$ if and only if $\rho^n \in \Omega_1^{\nu} \cup \Omega_2^{\nu}$. The sets $\Pi^{\nu}$ are compact. Moreover, we may write $\Pi(\cSo) = \bigcap_{\nu > 0} \Pi^{\nu}$ to recover the definition in~\eqref{eq:pi}.

Now, we will choose the parameters $\xi$ and $\nu$ differently depending on some properties of $\cSo$. Let us first consider two cases for which $v_{\eps}(\cSo) > 0$.

\begin{enumerate}
\item $v_{\min}(\cSo) >0$. In this case, the constant $\xi >0$ is chosen to be $\xi = \frac{ v_{\min}(\cSo)}2 > 0$. Now, for all $\nu$ sufficiently small we have $\inf_{\bbP \in \Pi^{\nu}} V(\bbP|\sigma) > \xi$ so that
$\Omega_1^{\nu} \cap \Omega_2^{\nu} \setminus \Omega_3^{\xi}$ is empty. Thus, combining \eqref{eq:outliers} and \eqref{eq:berry1}, we find
\begin{align}
  \chi_h^{\eps_n}(\cSo^{\otimes n}) \leq \sup_{\rho^n \in \cSo^{\otimes n}} D_h^{\eps_n}(\rho^n \| \sigma^n) \leq n \chi(\cSo) + \sqrt{n\,v_{\eps}^{\nu}(\cSo)}\,\Phi^{-1}(\eps) + O(\log n) \label{eq:type1}
\end{align}
 \item $\eps > \frac{1}{2}$ and $v_{\max}(\cSo)>v_{\min}(\cSo)=0$. Here we note that $\Phi^{-1}(\eps) > 0$ and $v_\eps(\cSo)>0$. Thus, we may choose $\xi>0$ sufficiently small so that 
\begin{equation}
\sqrt{\frac{\xi}{\eps^*}}\le \sqrt{v_\eps(\cSo)}\Phi^{-1}(\eps) \leq \sqrt{v_\eps^{\nu}(\cSo)}\Phi^{-1}(\eps) \nonumber 
\end{equation}
for any $\nu > 0$.
The bounds~\eqref{eq:outliers},~\eqref{eq:cheby1} and~\eqref{eq:berry1} can then be summarized and~\eqref{eq:type1} holds. 
\end{enumerate}
The bounds for cases 1 and 2 can be restated as follows. For all $\nu > 0$ sufficiently small, we have
\begin{align}
  \limsup_{n \to \infty} \frac{\chi_h^{\eps_n}(\cSo^{\otimes n}) - n \chi(\cSo) }{\sqrt{n}} \leq \sqrt{v_{\eps}^{\nu}}\,\Phi^{-1}(\eps)
\end{align}
Since $\nu > 0$ is arbitrary small we take $\nu \searrow 0$. Then, it remains to show that $\lim_{\nu \to 0} v_{\eps}^{\nu}(\cSo) = v_{\eps}(\cSo)$. This is a consequence of the following lemma (proved in Appendix~\ref{app:v-limit}).
\begin{lemma} \label{lm:v-limit}
  Let $\Theta_1 \supseteq \Theta_2 \supseteq \ldots$ be a sequence of compact sets in a metric space and let $f: \Theta_1 \to \mathbb{R}$ be continuous and bounded. Then,
  \begin{align}
     \lim_{n \to \infty}\, \inf_{x \in \Theta_n} f(x) = \inf_{x \in \Theta_{\infty}} f(x)  \qquad \textrm{whenever}\qquad \Theta_{\infty} := \bigcap_{n \in \mathbb{N}} \Theta_n \neq \emptyset. \label{eq:limitcompact}
  \end{align}
\end{lemma}

\noindent This establishes that $\chi_h^{\eps_n}(\cSo^{\otimes n}) \leq n \chi(\cSo) + \sqrt{n\, v_{\eps}(\cSo)}\,\Phi^{-1}(\eps) + o(\sqrt{n})$, as desired.

Let us now turn our attention to the cases for which $v_{\eps}(\cSo) = 0$.

\begin{enumerate}

\item[3.] $v_{\max}(\cSo) = v_{\min}(\cSo) = 0$. Here, we note that for any $\xi > 0$ there exists a $\nu > 0$ such that $\sup_{\bbP \in \Pi^{\nu}} V(\bbP|\sigma) < \xi$ and, thus, the set $\Omega_1^{\nu} \cap \Omega_2^{\nu} \cap \Omega_3^{\xi}$ is empty. Hence, the bounds~\eqref{eq:outliers} and~\eqref{eq:cheby1} can be combined to yield
\begin{align}
  \chi_h^{\eps_n}(\cSo^{\otimes n}) \leq \sup_{\rho^n \in \cSo^{\otimes n}} D_h^{\eps_n}(\rho^n \| \sigma^n) \leq n \chi(\cSo) + \sqrt{\frac{n \xi}{\eps^*}} + O(\log n) \label{eq:type2} .
\end{align}

\item[4.] $\eps \leq \frac{1}{2}$ and $v_{\max}(\cSo) > v_{\min}(\cSo) = 0$. Here, any choice of $\xi>0$ enforces that
\begin{align}
  \sqrt{\frac{\xi}{\eps^*}} \geq 0 = \sqrt{v_{\eps}^{\nu}(\cSo)}\,\Phi^{-1}(\eps) . 
\end{align}
Thus, the bounds~\eqref{eq:outliers},~\eqref{eq:cheby1} and~\eqref{eq:berry1} together establish that~\eqref{eq:type2} holds.

\end{enumerate}

Again, let us restate the bounds for cases 3 and 4 as follows. For all $\xi > 0$, we have
\begin{align}
\limsup_{n \to \infty}  \frac{\chi_h^{\eps_n}(\cSo^{\otimes n}) - n \chi(\cSo) }{\sqrt{n}} \leq \sqrt{ \frac{\xi}{\eps^*}}.\nonumber
\end{align}
Since $\xi>0$ is arbitrary, we may take $\xi \searrow 0$ and deduce  that 
\begin{equation}
\chi_h^{\eps_n}(\cSo^{\otimes n}) \le  n\chi(\cSo) + o\big(\sqrt{n}\big).\nonumber
\end{equation}
This concurs with the second-order approximation  since $v_\eps(\cSo)$ is zero and concludes the proof.
\end{proof}

\subsection{Asymptotics of the $\eps$-Hypothesis-Testing Divergence Radius: Beyond Second-Order}
\label{sec:converethird}

In this section we want to improve the upper bound of $o(\sqrt{n})$ in Theorem~\ref{th:main} to $O(\log n)$ for the important special case where $\cSo$ is a discrete set. To simplify the exposition here, we further assume that $v_{\min}(\cSo) > 0$. Comparing with the proof of Proposition~\ref{pr:thedifficultpart}, it is however easy to see that this condition can be relaxed to $v_{\eps}(\cSo) > 0$. 

\begin{proposition}
  Let $\eps \in (0,1)$ and $\cSo \subseteq \cS$ be discrete and $v_{\min}(\cSo) > 0$. Let $\{\eps_n\}_{n=1}^{\infty}$ be any sequence satisfying $|\eps_n - \eps| \leq 1/\sqrt{n}$ for all $n$. Then,
  \begin{align}
    \chi_h^{\eps_n}(\cSo^{\otimes n}) \leq n\, \chi(\cSo) + \sqrt{n\, v_{\eps}(\cSo)}\, \Phi^{-1}(\eps) + O(\log n) .
  \end{align}
\end{proposition}

\begin{proof}
  For any $n$, consider all sequences $\rho^n = \bigotimes_{i=1}^n \rho_i$ with $\rho_i \in \cSo$. The method of types~\cite{csiszar98} reveals that $P_{\rho^n}$ is in a set $\cP_n(\cSo) \subseteq \cP(\cSo)$ with cardinality satisfying $\log |\cP_n(\cSo)| = O(\log n)$.
  
  We use this for a further refinement of our state $\sigma^n$ (see also~\cite[Sec.~X.A]{hayashi09}) as follows:
\begin{align}
\sigma^n := \frac13 \sigma^{\otimes n} + \frac1{3 |\cG^\gamma|} \sum_{\tau \in   \cG^\gamma } \tau^{\otimes n} + \frac{1}{3 |\cP_n(\cSo)|} \sum_{P \in \cP_n(\cSo)} \Big( \rho^{(P)} \Big)^{\otimes n} . \label{eq:sigman2}
\end{align}
  Clearly, Proposition~\ref{pr:conv-bad} still applies with this definition, and for any $\nu \in (0, 1)$ we find that
  \begin{align}
    \sup_{\rho^n \notin \Omega_1^{\nu} \cap \Omega_2^{\nu}} D_h^{\eps_n}(\rho^n \| \sigma^n) 
       \leq n\, \chi(\cSo) + \sqrt{n\,v_{\eps}(\cSo)}\,\Phi^{-1}(\eps)
  \end{align}
  Now, observe that due to our condition on the channel, we have $v_{\min}(\cSo) > 0$. Thus, $V(P) = \sum P(\rho)\,V\big(\rho\big\|\rho^{(P)}\big)$ evaluated for $P \in \Pi$ is lower bounded by $v_{\min}(\cSo)$. Moreover, by continuity, $\inf_{P \in \Pi^{\nu}} V(P) > v_{\min}(\cSo)/2$ if $\nu$ is chosen sufficiently small.
  Thus, we in particular have that $V(P_{\rho^n}) > v_{\min}/2$ for all $\rho^n \in \Omega_1^{\nu} \cap \Omega_2^{\nu}$. For such a sequence $\rho^n$, we apply Proposition~\ref{pr:asymptotics} to find
  \begin{align}
    D_h^{\eps_n}(\rho^n\|\sigma^n) &\leq D_h^{\eps}\Big(\rho^n\Big\|\big(\rho^{(P)}\big)^{\otimes n}\Big) + \log |\cP_n(\cSo)|\\
      &\leq n I(P_{\rho^n}) + \sqrt{n\, V(P_{\rho^n})}\,\Phi^{-1}(\eps) + \log |\cP_n(\cSo)| + L_3 \log n
  \end{align}
  for $n \geq N_3$. Thus, we immediately find
  \begin{align}
    \chi_h^{\eps_n}(\cSo^{\otimes n}) \leq  \sup_{P \in \Pi^{\nu}} \bigg( n I(P) + \sqrt{n\,V(P)}\,\Phi^{-1}(\eps) \bigg) + O(\log n) \label{eq:takethis1}
  \end{align}
  and it only remains to show that the supremum is achieved in $\Pi$, without too much loss. As Polyanskiy, Poor and Verd\'u discuss in~\cite[App.~J]{polyanskiy10}, we indeed have
  \begin{align}
    \sup_{P \in \Pi^{\nu}} \bigg( n I(P) + \sqrt{n\,V(P)}\,\Phi^{-1}(\eps) \bigg) &= \sup_{P \in \Pi} \bigg( n I(P) + \sqrt{n\,V(P)}\,\Phi^{-1}(\eps) \bigg) + O(1) \label{eq:takethis2}
  \end{align}
  if $I(P)$ drops fast enough when we move away from $\Pi$ (but stay in $\Pi^{\nu}$) in the following sense. We require that $\frac{\diff^2}{\diff \alpha^2} I(P + \alpha v) \big|_{\alpha=0}$ is strictly negative for all $P \in \Pi$ and for all vectors $v$ satisfying $\sum_{\rho \in \cSo} v(\rho) = 0$ such that $P + v \notin \Pi$.\footnote{Note that $\frac{\diff}{\diff \alpha} I(P + \alpha v) \big|_{\alpha=0} = 0$ on $\Pi$ by definition as $I(P)$ is maximized on $\Pi$.} This is equivalent to the condition 
  \begin{align}
  \frac{\diff^2}{\diff \alpha^2} H\Big(\rho^{(P)} + \alpha \Delta^{(v)} \Big) \bigg|_{\alpha = 0} < 0, \qquad \textrm{where} \quad \Delta^{(v)} := \sum_{\rho \in \cSo} v(\rho) \rho ,
  \end{align}
which is satisfied due to Lemma~\ref{lm:strict} below.

Thus, combining~\eqref{eq:takethis1} and~\eqref{eq:takethis2}, we conclude that
\begin{align}
  \chi_h^{\eps_n}(\cSo^{\otimes n}) &\leq  \sup_{P \in \Pi} \bigg( n I(P) + \sqrt{n\,V(P)}\,\Phi^{-1}(\eps) \bigg) + O(\log n) \\
    &= n \chi(\cSo) + \sqrt{n\,v_{\eps}(\cSo)}\,\Phi^{-1}(\eps) + O(\log n) \,.
\end{align}
\end{proof}

\begin{lemma}
  \label{lm:strict}
  Let $\rho \in \cS$ and let $\Delta \in \cH$ with $\tr(\Delta)=0$, $\Delta \neq 0$ and $\Delta \ll \rho$. Then, we have $\frac{\diff^2}{\diff \lambda^2} H(\rho + \lambda \Delta) \big|_{\lambda=0} < 0$.
  In particular, $\rho \mapsto H(\rho)$ is strictly concave.
\end{lemma}

Note that strict negativity of the second derivative is a stronger property than strict concavity, which it implies.\footnote{This is revealed, for example, by the behavior of the function $t \mapsto -t^4$ at $t = 0$.}

We are grateful to David Reeb for allowing us to present a proof based on his ideas here~\cite{reebpc13}.

\begin{proof}
  We define $\rho_{\lambda} := \rho + \lambda \Delta$.
  First, we note that $\frac{\diff}{\diff \lambda} \rho_{\lambda}^{-1} = - \rho_{\lambda}^{-1} \big( \frac{\diff}{\diff \lambda} {\rho_{\lambda}} \big) \rho_{\lambda}^{-1} = - \rho_{\lambda}^{-1} \Delta \rho_{\lambda}^{-1}$ by applying the product rule to $\frac{\diff}{\diff \lambda} \big( \rho_{\lambda}^{-1} \rho_{\lambda} \big)$. 
  Since $\Delta \ll \rho$, we can restrict to the subspace $\{\rho > 0\}$ without loss of generality.
  There, for $\lambda$ small enough such that $\rho_{\lambda} > 0$, we use the integral representation
  \begin{align}
   \log \rho_{\lambda} = \int_0^\infty \diff s\ (1-s)^{-1} \id - (\rho_{\lambda} + s\,\id)^{-1}
   \end{align}
which directly follows from its scalar analogue. As such, it is easy to compute
\begin{align}
  \frac{\diff}{\diff \lambda} \log \rho_{\lambda} = \int_0^\infty \diff s - \frac{\diff}{\diff \lambda} (\rho_{\lambda} + s\,\id)^{-1} = \int_0^\infty \diff s\, (\rho_{\lambda} + s\,\id)^{-1} \Delta  (\rho_{\lambda} + s\,\id)^{-1} .
\end{align}
  Recalling that $\frac{\diff}{\diff \lambda} \tr\big(f(\rho_{\lambda})\big) = \tr\big(f'(\lambda) \frac{\diff}{\diff \lambda} \rho_{\lambda} \big)$ for $f: t \mapsto - t \log t$, we find that
  \begin{align}
    \frac{\diff}{\diff \lambda} H(\rho_{\lambda}) &= - \tr(\Delta \log \rho_{\lambda}) \qquad 
    \qquad \textrm{and}\\
    \frac{\diff^2}{\diff \lambda^2} H(\rho_{\lambda}) &= - \int_0^\infty \diff s \tr \Big( \Delta (\rho_{\lambda} + s\,\id)^{-1} \Delta (\rho_{\lambda} + s\,\id)^{-1} \Big) \\
    &= - \int_0^\infty \diff s \left\| (\rho_{\lambda} + s\,\id)^{-\frac12} \Delta (\rho_{\lambda} + s\,\id)^{-\frac12} \right\|_2^2.
  \end{align}
  For all $s > 0$ we thus find that the integrand is positive whenever $(\rho_{\lambda} + s\,\id)^{-\frac12} \Delta (\rho_{\lambda} + s\,\id)^{-\frac12} \neq 0$, which is evident since $\Delta \neq 0$ and $\rho_{\lambda} + s\,\id$ has full support. Hence, the desired inequality holds.
\end{proof}

\paragraph*{Acknowledgements:}
We thank Andreas Winter and Mark Wilde for discussions and comments on a previous version of this manuscript. MT also thanks David Reeb, Mil\'an Mosonyi and especially Corsin Pfister for many insightful discussions, and the Isaac Newton Institute (Cambridge) for its hospitality while part of this work was completed. MT acknowledges funding by the Ministry of Education (MOE) and National Research Foundation Singapore, as well as MOE Tier 3 Grant ``Random numbers from quantum processes'' (MOE2012-T3-1-009). VT gratefully acknowledges financial support from the National University of Singapore (NUS) under startup grants R-263-000-A98-750/133 and the NUS Young Investigator Award R-263-000-B37-133.

\appendix

\newcommand{\nn}{\nonumber}

\section{Proof of Lemma~\ref{lm:set-limit}}\label{app:set-limit}

\begin{proof}
The inclusion $\supseteq$ is obvious because by the monotonicity of the convex hull operator and the fact that $\Theta_n\supseteq \Theta_\infty$ for any $n\in\mathbb{N}$, we have 
$\conv(\Theta_n) \supseteq \conv(\Theta_\infty )$.

It remains to prove the inclusion  $\subseteq$. Let 
\begin{align}
\rho \in \bigcap_{n\in\mathbb{N}}\conv( \Theta_n) .
\end{align}
 This means that for every $n\in\mathbb{N}$, $\rho$ can be written as 
$\rho = \sum_{j=1}^{l} \alpha_{jn} \rho_{jn}$ where $\rho_{jn} \in \Theta_n$ for each $j = 1,\ldots, \ell$ and $(\alpha_{1n},\ldots,\alpha_{\ell n} )$ is a probability distribution. Note that $\ell$ is finite and does not depend on $n$ due to Caratheodory's theorem since $\Theta_n$ for each $n$ are subsets of the same finite-dimensional vector space.

Consider the sequence $\{\rho_{1n}\}_{n\in\mathbb{N}} \subset \Theta_1$, i.e., $j=1$. Since $\Theta_1$ is compact, there must exists a convergent subsequence, say indexed by $n_{k}[1]$, i.e., the sequence $\{\rho_{1n_{k}{[1]} }\}_{k\in\mathbb{N}}$ is convergent and
 \begin{equation}
\lim_{k\to\infty} \rho_{1n_{k}{[1]}}=\rho_1
 \end{equation}
where $\rho_1 \in \Theta_\infty$ since $\Theta_n$ decrease to $\Theta_\infty$. Now, consider the sequence $\{ \rho_{2 n_k{[1]}}\}_{k\in\mathbb{N}}$.  By the same argument, we may extract a subsequence of $n_{k}{[1]}$ indexed by $n_k{[2]}$ for which 
\begin{equation}
\lim_{k\to\infty} \rho_{2n_{k}{[2]}}=\rho_2, \qquad \textrm{and} \quad \rho_2 \in \Theta_\infty  \nn .
\end{equation}
Continue extracting subsequences until we reach $\ell$.  Now consider the subsequence  indexed by
$m_k := n_k{[\ell]}$.
Clearly, $\rho$  can be written also as 
\begin{equation}
\rho = \sum_{j=1}^{\ell} \alpha_{j m_k} \rho_{j m_k}.  \label{eqn:rep_rho}
\end{equation}
By construction, each $\rho_{j m_k}$ converges to $\rho_j \in \Theta_\infty$ when we take $k\to\infty$. So by representation of $\rho$ in \eqref{eqn:rep_rho},  and the arbitrariness of $k$, we have that $\rho$ is a convex combination of elements from $\Theta_\infty$, i.e., $\rho\in\conv(\Theta_\infty)$ as desired.
\end{proof}

\section{Background and Proof of Proposition~\ref{pr:asymptotics}}
\label{app:hypo}

\subsection{Nussbaum-Sko\l{}a Distributions}

For the proof we leverage on a hierarchy of information measures in quantum information that was introduced in~\cite{tomamichel12}. 
To apply these results, let us first review the following concept. For any two quantum states $\rho, \sigma \in \cS$, we define their (classical) \emph{Nussbaum-Sko\l{}a distributions} $P^{\rho,\sigma},\, Q^{\rho,\sigma} \in \cP\big([d] \times [d]\big)$ via the relations~\cite{nussbaum09}
\begin{align}
  P^{\rho,\sigma}(a,b) = r_a \big| \langle \phi_a | \psi_b \rangle \big|^2 \quad \textrm{and} \quad
  Q^{\rho,\sigma}(a,b) = s_b \big| \langle \phi_a | \psi_b \rangle \big|^2\,,
\end{align}
where $\rho = \sum_a r_a \proj{\phi_a}$ and $\sigma = \sum_b s_b \proj{\psi_b}$. We summarize some properties of the Nussbaum-Sko\l{}a distributions that will turn out to be of great use in the sequel (these were already pointed out in~\cite{tomamichel12}). 
First, it is easy to verify  by substitution that
\begin{align}
 \label{eq:nussbaum-moments}
  D(\rho\|\sigma) = D(P^{\rho,\sigma} \| Q^{\rho,\sigma} ) \quad \textrm{and} \quad V(\rho\|\sigma) = V(P^{\rho,\sigma} \| Q^{\rho,\sigma} ) \,.
\end{align}
Second, for product states $\rho_1 \otimes \rho_2$ and $\sigma_1 \otimes \sigma_2$, we have
\begin{align}
  \label{eq:nussbaum-product}
  P^{\rho_1 \otimes \rho_2,\sigma_1 \otimes \sigma_2} = P^{\rho_1,\sigma_1} \otimes P^{\rho_2,\sigma_2}, \quad \textrm{and}
  \quad Q^{\rho_1 \otimes \rho_2,\sigma_1 \otimes \sigma_2} = Q^{\rho_1,\sigma_1} \otimes Q^{\rho_2,\sigma_2} \,.
\end{align}
Third, the condition $\sigma \gg \rho$ holds if and only if $Q^{\rho,\sigma} \gg P^{\rho,\sigma}$.
Now, let 
\begin{align}
\Xi(\sigma) := 2 \Big\lceil \log \frac{\lambda_{\max}(\sigma)}{\tilde{\lambda}_{\min}(\sigma)} \Big\rceil,
\end{align}
where $\lambda_{\max}(\sigma)$ and $\tilde{\lambda}_{\min}(\sigma)$ denote the largest and smallest nonzero eigenvalues of $\sigma$, respectively.
\begin{lemma} \textnormal{{\cite[Thm.~14]{tomamichel12}}} \label{lm:q-to-cl}
  Let $\rho, \sigma \in \cS$ and $\sigma \gg \rho$.  Then, for $0 < \delta < \min\{\eps,\frac{1-\eps}4\}$,
  \begin{align}
   D_h^{\eps}(\rho\|\sigma) &\leq D_s^{\eps+4\delta}(P^{\rho,\sigma} \| Q^{\rho,\sigma})  
   + \log \Xi(\sigma) + 4\log \frac{1}{\delta} + F_1(\eps,\delta)   
   \,, \quad \textrm{and} \\
   D_h^{\eps}(\rho\|\sigma) &\geq D_s^{\eps-\delta}(P^{\rho,\sigma} \| Q^{\rho,\sigma}) 
   - \log \Xi(\sigma) - \log \frac{1}{\delta} - F_2(\eps) \,,
   \end{align}
   where $F_1(\eps,\delta) := \log \frac{(1-\eps) (\eps+3\delta)}{1-(\eps+3\delta)}$ and $F_2(\eps) := \log \frac{1}{1-\eps}$.
\end{lemma}
Here, the (classical) \emph{information spectrum divergence}~(in the spirit of Verd\'u and Han~\cite{verdu93, han02}) for two probability distributions $P, Q \in \cP(\mathcal{X})$ (where $\mathcal{X}$ is a discrete set) is given by
\begin{align}
  D_s^{\eps}(P \| Q) := \sup \bigg\{ R \in \mathbb{R} \,\bigg|\, \Pr_{X \leftarrow P} \bigg[ \log \frac{P(X)}{Q(X)} \leq R \bigg] \leq \eps \bigg\}
  \label{eq:inf-spec} .
\end{align}

In the following, we also need the {\em third absolute moment of the log-likelihood ratio} between $P$ and $Q$, given as\footnote{It is not evident how a non-commutative version of this quantity should be defined directly; however, the commutative case is sufficient for our work.}
\begin{align}
  T(P\|Q) &:= \sum_{x\in\cX} P(x) \bigg| \log \frac{P(x)}{Q(x)} - D(P\|Q) \bigg|^3 \, \quad \textrm{and} \quad T(\rho\|\sigma) := T\big(P^{\rho,\sigma}\big\|Q^{\rho,\sigma}\big) \,. 
\end{align}

\subsection{Non-Asymptotic Bounds on the $\eps$-Hypothesis Testing Divergence}

It is immediate that the probability appearing in the definition of the information spectrum divergence evaluated for product distributions is subject to the central limit theorem if the variance of $\log \frac{P}{Q}$ is bounded away from zero.

\begin{lemma}\label{lm:asymptotics-2}
  Let $ n \geq 1$, $\{ \rho_i \}_{i=1}^n$, for $\rho_i \in \cS$ a set of states and let $\sigma \in \cS$ such that $\sigma \gg \rho_i$ for all $i \in [n]$. Moreover, let $\eps \in (0,1)$ and $\delta < \min\{\eps, \frac{1-\eps}{4}\}$. Define
  \begin{align}
    D_n := \frac{1}{n} \sum_{i=1}^n D(\rho_i \| \sigma), \quad
    V_n := \frac{1}{n} \sum_{i=1}^n V(\rho_i \| \sigma), \quad
    T_n := \frac{1}{n} \sum_{i=1}^n T(\rho_i \| \sigma).
  \end{align}
  Then, the following  Chebyshev-type inequalities hold:
  \begin{align}
    D_h^{\eps}\Big(\bigotimes_{i=1}^n \rho_i \Big\| \sigma^{\otimes n} \Big) &\leq n D_n + \sqrt{\frac{n V_n}{1-\eps-4\delta}} + \log \big(n \Xi(\sigma)\big) + 4 \log \frac{1}{\delta} + F_{1}(\eps,\delta) \,,\nonumber\\
    D_h^{\eps}\Big(\bigotimes_{i=1}^n \rho_i \Big\| \sigma^{\otimes n} \Big) &\geq n D_n - \sqrt{\frac{n V_n}{\eps - \delta}} - \log \big(n \Xi(\sigma)\big) - \log \frac{1}{\delta} - F_{2}(\eps) \,. \label{eq:cheby-2}
  \end{align}
  Moreover, if $V_n > 0$, then the following Berry-Esseen type bounds holds:
  \begin{align}
    D_h^{\eps}\Big(\bigotimes_{i=1}^n \rho_i \Big\| \sigma^{\otimes n} \Big) &\leq n D_n + \sqrt{n V_n} \Phi^{-1}\bigg( \eps + 4\delta + \frac{6\, T_n}{\sqrt{nV_n^3}} \bigg) + \log \big(n \Xi(\sigma)\big)+ 4 \log \frac{1}{\delta} + F_{1}(\eps,\delta) \,,\nonumber\\
    D_h^{\eps}\Big(\bigotimes_{i=1}^n \rho_i \Big\| \sigma^{\otimes n} \Big) &\geq n D_n + \sqrt{n V_n} \Phi^{-1}\bigg( \eps - \delta - \frac{6\, T_{n}}{\sqrt{nV_n^3} } \bigg) - \log \big(n \Xi(\sigma)\big)- \log \frac{1}{\delta} - F_{2}(\eps)  \,,  \label{eq:berry-2}
  \end{align}
  where   $F_1, F_2$ are given in Lemma~\ref{lm:q-to-cl}.
\end{lemma}

\begin{proof}[Proof (Sketch)]
  We first apply Lemma~\ref{lm:q-to-cl} to replace $D_h^{\eps}$ with $D_s^{\eps+4\delta}$ (for the upper bounds) and $D_s^{\eps-\delta}$ (for the lower bound). For this purpose, we note that $\Xi(\sigma^{\otimes n}) \leq n \Xi(\sigma)$. For the upper bound, this yields
  \begin{align}
     D_h^{\eps}\Big(\bigotimes_{i=1}^n \rho_i \Big\| \sigma^{\otimes n} \Big) \leq
     D_s^{\eps+4\delta} \Big( \bigotimes_{i=1}^n P^{\rho_i,\sigma} \Big\| \bigotimes_{i=1}^n Q^{\rho_i,\sigma} \Big)  
   + \log \big(n \Xi(\sigma)\big)+ 4\log \frac{1}{\delta} + F_1(\eps,\delta) \label{eq:asym-proof1}
  \end{align}
  Note that the information spectrum divergence on the right-hand side of~\eqref{eq:asym-proof1} is evaluated for classical product distributions $\bigotimes_{i=1}^n P^{\rho_i,\sigma}$ and $\bigotimes_{i=1}^n Q^{\rho_i,\sigma}$. Consider the independent random variables
  \begin{align}
  Z_i := \log \frac{P^{\rho_i,\sigma}(A_i, B_i)}{Q^{\rho_i,\sigma}(A_i, B_i)}, \qquad (A_i, B_i) \leftarrow P^{\rho_i,\sigma}
  \end{align}
  for each $i \in [n]$. Then, the definition of the information spectrum divergence in~\eqref{eq:inf-spec} yields
  \begin{align}
    D_s^{\eps+4\delta} \Big( \bigotimes_{i=1}^n P^{\rho_i,\sigma} \Big\| \bigotimes_{i=1}^n Q^{\rho_i,\sigma} \Big) = \sup \bigg\{ R \in \mathbb{R} \,\bigg|\, \Pr \bigg[ \sum_{i=1}^n Z_i \leq R \bigg] \leq \eps + 4\delta \bigg\}  \,. \label{eq:asym-proof2}
  \end{align}
  Further, observe that the average mean and variance of $Z_i$ are respectively given by
  \begin{align}
    \frac{1}{n} \sum_{i=1}^n \Exp [Z_i] = \frac{1}{n} \sum_{i=1}^n D(P^{\rho_i,\sigma} \| Q^{\rho_i,\sigma}) &= \frac{1}{n} \sum_{i=1}^n D(\rho_i\|\sigma) = D_n \,, \qquad \textrm{and} \\
    \frac{1}{n} \sum_{i=1}^n \Var [Z_i] = \frac{1}{n} \sum_{i=1}^n V(P^{\rho_i,\sigma} \| Q^{\rho_i,\sigma}) &= \frac{1}{n} \sum_{i=1}^n V(\rho_i\|\sigma) = V_n \, .
  \end{align}
  Thus, we apply standard Chebyshev or Berry-Esseen~\cite[Sec.\ XVI.5]{feller71} bounds on the probability in~\eqref{eq:asym-proof2}. (See, e.g.~\cite[Lem.~5]{tomamicheltan12}, for details.)
  The proof of the lower bounds proceeds analogously.
\end{proof}

\subsection{Uniform Upper Bounds}

The following two lemmas give uniform upper bounds on $V(\rho\|\sigma)$ and $T(\rho\|\sigma)$.

\begin{lemma}
  \label{lm:V-uni}
  Let $\cSo \subset \cS$ and $\lambda_0 > 0$. Then, there exists a constant $V^+(\cSo,\lambda_0)$ such that $V(\rho\|\sigma) \leq V^+$ for all $\rho \in \cSo$ and $\sigma \in \cS$ such that $\lambda_{\min}(\sigma) \geq \lambda_0$.
\end{lemma}

\begin{proof}
  First, note that $(\rho, \sigma) \mapsto V(\rho\|\sigma)$ is continuous on the compact set $\ccSo \times \{ \sigma \in \cS \,|\, \lambda_{\min}(\sigma) \geq \lambda_0 \}$ since $\sigma \gg \rho$ everywhere. Thus, we may simply choose
  \begin{align}
    V^+ &:= \max \big\{ V(\rho\|\sigma) \,\big|\, \rho \in \ccSo,\ \sigma \in \cS,\ \lambda_{\min}(\sigma) \geq \lambda_0 \big\} . \qedhere
   \end{align}
\end{proof}

\begin{lemma}
  \label{lm:T-uni}
  Let $\cSo \subset \cS$ and $\sigma \in \cS$ such that $\lambda_{\min}(\sigma) > 0$. Then, there exists a constant $T^+(\cSo,\sigma)$ such that $T(\rho\|\sigma) \leq T^+$ for all $\rho \in \cSo$.
\end{lemma}

\begin{proof}
  We have $\sigma \gg \rho$ and thus $Q^{\rho,\sigma} \gg P^{\rho,\sigma}$ for all $\rho \in \cSo$ since $\sigma$ is strictly positive. Hence, $\rho \mapsto T(\rho\|\sigma) = T(P^{\rho,\sigma}\|Q^{\rho,\sigma})$ is continuous and it suffices to define $T^+ := \max_{\rho \in \ccSo} T(\rho\|\sigma)$.
\end{proof}

For the following, let us define $V(P) := V\big(P\big|\rho^{(P)}\big)$ and
$T(P) := \sum_{\rho \in \cSo} P(\rho)\, T\big(\rho\big\|\rho^{(P)}\big)$ for all $P \in \cP(\cSo)$ in a discrete set $\cSo$. These quantitates have the following uniform upper bounds:

\begin{lemma} \label{lm:VT-uni2}
  Let $\cSo \subset \cS$ be discrete. Then, there exist constants $V^{*}(\cSo)$ and $T^{*}(\cSo)$ such that
   $V(P) \leq V^{*}$ and $T(P) \leq T^{*}$ for all $P \in \cP(\cSo)$.
\end{lemma}

\begin{proof}
  To convince ourselves that the functions $P \mapsto V(P)$ and $P \mapsto T(P)$ are continuous, we note that, for all $P \in \cP(\cSo)$ and all $\rho \in \cSo$ at least one of the following conditions holds 1) $P(\rho) = 0$ or 2) $\rho^{(P)} \gg \rho$. The lemma then follows from the fact that $\cP(\cSo)$ is compact.
\end{proof}

\subsection{Proof of Proposition~\ref{pr:asymptotics}}

\begin{proof}[Proof of Proposition~\ref{pr:asymptotics}]
  The first statement relies on the Chebyshev-type inequalities in~\eqref{eq:cheby-2} in Lemma~\ref{lm:asymptotics-2}, which for any $\delta = \frac{1}{\sqrt{n}}$ and for $n$ sufficiently large such that $\frac{2}{\sqrt{n}} < \min \{ \eps, \frac{1-\eps}{4} \}$ yield
  \begin{align}
    \bigg| D_h^{\eps_n}\Big(\bigotimes_{i=1}^n \rho_i \Big\| \sigma^{\otimes n} \Big) - n D_n \bigg| \leq \sqrt{\frac{n V_n}{\min\big\{ 1-\eps_n-\frac{4}{\sqrt{n}}, \eps_n - \frac1{\sqrt{n}}\big\}}} + 3 \log n + \log \Xi(\sigma) \\
    \qquad \qquad + \max \Big\{ F_1\Big(\eps_n,\frac{1}{\sqrt{n}}\Big), F_2(\eps_n) \Big\} .
  \end{align}
  Now, we note that $\Xi(\sigma) \leq 2 \log \frac{1}{\lambda_0} + 1 = O(1)$ and note that
  \begin{align}
    &\sqrt{\frac{n V_n}{\min\big\{ 1-\eps_n-\frac{4}{\sqrt{n}}, \eps_n - \frac1{\sqrt{n}}\big\}}} \leq \sqrt{\frac{n V_n}{\min\big\{ 1-\eps-\frac{5}{\sqrt{n}}, \eps - \frac2{\sqrt{n}}\big\}}} = \sqrt{\frac{n V_n}{\eps^*}} + O(1) \qquad \textrm{and} \\
     &\max \Big\{ F_1\Big(\eps_n,\frac{1}{\sqrt{n}}\Big), F_2(\eps_n) \Big\} \leq \max \Big\{ F_1\Big(\eps+\frac{1}{\sqrt{n}},\frac{1}{\sqrt{n}}\Big), F_2(\eps-\frac{1}{\sqrt{n}}) \Big\} = O(1) .
  \end{align}
  Thus, any choice of $L_1 > 3$ will yield the desired result.
  
   Finally, we have $V_n \leq V^+$ by Lemma~\ref{lm:V-uni} due to the assumption on $\lambda_{\min}(\sigma)$. We can thus pick the constant 
   \begin{align}
   K_1(\eps,\cSo,\lambda_0) > \sqrt{\frac{V^+}{\eps^*}}
   \end{align}
    uniformly in $\{ \rho_i \}_{i=1}^n$. Finally, for any such choices of $L_1$ and $K_1$, we find a number $N_1(\eps,\cSo,\lambda_0)$ such that the statement holds.
   
  The second statement is based on the Berry-Esseen-type inequalities in~\eqref{eq:berry-2} in Lemma~\ref{lm:asymptotics-2}. We prove the upper bound and note that the lower bound follows by an analogous argument. First, we use~\eqref{eq:berry-2} and set $\delta = \frac{1}{\sqrt{n}}$ to establish that
  \begin{align}
D_h^{\eps_n}\Big(\bigotimes_{i=1}^n \rho_i \Big\| \sigma^{\otimes n} \Big) &\leq n D_n + \sqrt{n V_n} \Phi^{-1}\bigg( \eps_n + \frac{4}{\sqrt{n}} + \frac{6\, T_n}{\sqrt{nV_n^3}} \bigg) + 3 \log n + \log \Xi(\sigma) + F_{1}\Big(\eps_n,\frac1{\sqrt{n}}\Big)  
\end{align}
  Now, note that $V_n \geq \xi$ by assumption of the theorem and $T_n \leq T^+(\cSo,\sigma)$ by Lemma~\ref{lm:T-uni}. Since $\eps \mapsto \Phi^{-1}(\eps)$ is monotonically increasing, we find
  \begin{align}
    \Phi^{-1}\bigg( \eps_n + \frac{4}{\sqrt{n}} + \frac{6\, T_n}{\sqrt{nV_n^3}} \bigg) \leq \Phi^{-1}\bigg(\eps + \frac{B}{\sqrt{n}} \bigg), \qquad \textrm{where} \quad B = 5 + 6 \frac{T^+(\cSo, \sigma)}{\xi^{\frac{3}{2}}} .
  \end{align}
  Moreover, since $V_n \leq V^+$ and $\eps \mapsto \Phi^{-1}(\eps)$ is continuously differentiable we find that
  \begin{align}
    \sqrt{n V_n} \Phi^{-1}\bigg( \eps + \frac{B}{\sqrt{n}} \bigg) \leq \sqrt{n V_n} \Phi^{-1}(\eps) + O(1).
  \end{align}
  by Taylor's theorem.
  Collecting the remaining terms as $3 \log n + O(1)$ and choosing $L_2 > 3$ reveals that there exists a constant $N_2(\eps,\cSo,\sigma,\xi)$ such that the statement holds.
  
  To confirm the final statement, we need to be a bit more careful because $\lambda_{\min}(P_{\rho^n})$ can be arbitrarily close to zero and thus Lemmas~\ref{lm:V-uni} and~\ref{lm:T-uni} do not apply. However, the proof goes through analogously if we instead of these lemmas employ Lemma~\ref{lm:VT-uni2}.
\end{proof}

\section{Proof of Lemma~\ref{lm:net}}\label{app:net}

The following construction is likely not optimal in the parameters $\gamma$ and $|\cG^{\gamma}|$, but it suffices for our purpose and allows us to use previously established results.
\begin{proof}
  First, we employ a construction in~\cite[Lem.~II.4]{hayden04b} to establish that, for every $0 < \gamma < 1$, there exists a set of pure states $\{ \psi_i \}_{i \in [K]} \subseteq \cS_{\circ}$ with cardinality $K \leq (5/\gamma)^{2d}$ such that the following holds: for every $\phi \in \cS_{\circ}$, we have $\min_{i \in [K]} \| \phi - \psi_i \|_1 \leq \gamma$.
  
  Second, consider the set $\cP_m^{>0}$ of $m$-types~\cite{csiszar98} with full support, defined as
  \begin{align}
    \cP_m^{>0} := \big\{ P \in \cP([d]) \,\big|\, m P(i) \in [m] \textrm{ for all } i \in [d]  \big\} .
  \end{align}
  Setting $m = \lceil 2 d \frac{1}{\gamma} \rceil$, we will now show that, for every $P \in \cP([d])$, we have $\min_{Q \in \cP_m^{\setminus 0}} \| P - Q \|_1 \leq \gamma$. To see this, we construct a $Q \in \cP_m^{> 0}$ for every $P$ as follows. Start by setting $Q(i) = \frac{1}{m}$ for all $i \in [d]$. (Note that $m > d$ so that the total weight is smaller than one.) Then, pick any index $i$ for which $Q(i) < P(i)$ and increase $Q(i)$ by $\frac{1}{m}$. Repeat this until $Q$ is normalized.
   We observe that $\| P - Q \|_1 = 2 \sum_{i: Q(i) > P(i)} Q(i) - P(i) \leq \frac{2d}{m} \leq \gamma$ since $Q(i) - P(i)$ never exceeds $\frac{1}{m}$ by construction.
    Note that this choice also ensures that $\min_i Q(i) \geq \frac{1}{m}$.
  Furthermore, the number of types is bounded as~\cite{csiszar98}, 
  \begin{align}
  |\cP_m^{> 0}| 
  \leq (m+1)^{d-1} \leq ( 2d/\gamma + 2)^{d-1} .
  \end{align}
  
  Now, we are ready to define an $\gamma$-net for mixed states as follows:
  \begin{align}
    \cG^{\gamma} := \Big\{ \tau \in \cS \,\Big|\, \tau = \sum_{i=1}^d Q(i) \psi_{\ell(i)}, \ Q \in \cP_m^{> 0}, \ \ell: [d] \to [K] \Big\} .
  \end{align}
  We have $| \cG^{\gamma}| = K^d \cdot |\cP_m^{> 0}| \leq (5/\gamma)^{2d^2} (2d/\gamma+2)^{d-1}$. Moreover, let $\rho \in \cS(B)$ be an arbitrary state with $\rho = \sum_i P(i)\, \phi_i$ its eigenvalue decomposition, where $\phi_i \in \cS_{\circ}(B)$ are (mutually orthogonal) pure states and $P \in \cP(B)$. Now, choose $Q \in \cP_m^{> 0}$ and $\ell: [d] \to [K]$ such that
  \begin{align}
    \| P - Q \|_1 \leq \gamma \quad \textrm{and} \quad \forall i \in [d] : \| \psi_{\ell(i)} - \phi_i \|_1 \leq \eps .
  \end{align}
  For $\tau = \sum_{i=1}^d Q(i) \psi_{\ell(i)} \in \cG^{\gamma}$, we then have
   \begin{align}
     \| \rho - \tau \|_1 &\leq \sum_{i=1}^d \Big\| P(i) \phi_i - Q(i) \psi_{\ell(i)} \Big\|_1 \leq \sum_{i=1}^d  P(i) \Big\| \phi_i - \psi_{\ell(i)} \Big\|_1 + \big| P(i)  - Q(i) \big| \leq 2\gamma ,
   \end{align}
   where we used the triangle inequality multiple times.
   
   To get the second statement, we employ a continuity result by Audenaert and Eisert~\cite[Thm.~2]{audenaert05}, which ensures that
     $D(\rho\|\tau) \leq 4 \, \kappa^2/\beta$,   
     where $\beta$ is the minimum eigenvalue of $\tau$, and $\kappa := \frac12 \| \rho - \tau \|_1$. By our construction of $\tau$\,---\,in particular, recall the construction of $Q \in \cP_m^{> 0}$\,---\,we enforce that $\beta \geq \frac{1}{m}$. Hence, the above can be further bounded as
     \begin{align}
       D(\rho\|\tau) \leq 4 \kappa\cdot \frac{\kappa}{\beta} \leq 4\gamma (2d + 1) ,
     \end{align}
   where we used that $\kappa \leq \gamma$ and $\kappa/\beta \leq \gamma m = \gamma \lceil 2 d \frac{1}{\gamma} \rceil \leq 2d + 1$. 
   
   Finally, we note that every $\tau\in\cG^\gamma$ has minimum eigenvalue bounded from below by $\frac{1}{m}\ge 1/(2d/\gamma + 1)=\gamma/(2d+\gamma)$.  
\end{proof}

\section{Auxiliary Lemmas for Sections~\ref{sec:direct} and~\ref{sec:converse}}
\label{app:selection}
\label{app:v-limit}
\label{app:direct}

\subsection{Proof of Lemma~\ref{lm:v-expand}}

This is a straightforward generalization of the argument in~\cite[Lem.~62]{polyanskiy10}.

\begin{proof}
  By a simple calculation (or employing the law of total variance), it is easy to verify that
  \begin{align}
    &V\Bigg( \bigoplus_{\rho \in \cSo} P(\rho) \rho\, \Bigg\| \bigoplus_{\rho \in \cXo} P(\rho) 
    \rho^{(P)} \Bigg)\\
       &\qquad \qquad = \sum_{\rho \in \cSo} P(\rho)\, V\big(\rho\big\|\rho^{(P)}\big) + 
       \sum_{\rho \in \cXo} P(\rho) \bigg( D\big(\rho\big\|\rho^{(P)}\big) - \sum_{\rho \in \cSo} P(\rho)\, D\big(\rho\big\|\rho^{(P)}\big) \bigg)^2 .
  \end{align}
  Thus, if we choose $P \in \Pi(\cSo)$ we clearly have $\rho^{(P)} = \sigma^*(\cSo)$ and the second term vanishes due to Property 2 of Theorem~\ref{th:radius}.
\end{proof}

\subsection{Proof of Lemma~\ref{lm:selection}}

\begin{proof}
Let $\Xi^{\nu}$ be the set of indices for which $\Delta(\rho_i, \Gamma) >\frac{\nu}{2}$ holds. Then, we have
\begin{align}
\nu<\frac{1}{n}\sum_{i=1}^n\Delta(\rho_i,\Gamma) \le\frac{1}{n}\sum_{i\in\Xi^\nu}1+\frac{1}{n}\sum_{i\notin\Xi^\nu}\frac{\nu}{2}\le\frac{ | \Xi^\nu|}{n}+\frac{\nu}{2}
\end{align}
  from which the condition on the cardinality of $\Xi^{\nu}$ follows.
\end{proof}

\subsection{Proof of Lemma~\ref{lm:v-limit}}

\begin{proof}
  First note that all infima can be replaced with minima since the optimization is over compact sets. Denote $\min_{x \in \Theta_{\infty}} f(x)$ by $f*$.
  Clearly, $\limsup_{n \to \infty} \min_{x \in \Theta_n} f(x) \leq f^*$ since the inequality holds for every $n \in \mathbb{N}$ as $\Theta_n \supseteq \Theta_{\infty}$.
  
  Suppose, for the sake of contradiction that $\liminf_{n \to \infty} \min_{x \in \Theta_n} f(x) < f^*$. Then, there exists a subsequence indexed by $\{ n_k \}_{k \in \mathbb{N}}$ with the property that $\min_{x \in \Theta_{n_k}} f(x) < f^*$. For every $k \in \mathbb{N}$, let $x_k \in \argmin_{x \in \Theta_{n_k}} f(x)$ be any minimizer. Since the sets $\Theta_{n_k}$ are compact, there must exist a converging subsequence indexed by $\{ k_l \}_{l \in \mathbb{N}}$ such that $\lim_{l \to \infty} x_{k_l} = x^*$. Clearly, $x^*$ must be in $\Theta_{\infty}$. However, this leads to a contradiction with $f(x^*) < f^* = \min_{x \in \Theta_{\infty}} f(x)$.
  Hence, $\liminf_{n \to \infty} \min_{x \in \Theta_n} f(x) \geq f^*$.  
\end{proof}

\bibliographystyle{arxiv}
\bibliography{library}

\end{document}